\documentclass[runningheads]{llncs}
\usepackage{fullpage}
\usepackage[a4paper,left=2.2cm,right=2.2cm,top=2.5cm,bottom=2.8cm]{geometry}
\usepackage[numbers,sort&compress]{natbib}
\usepackage[T1]{fontenc}
\usepackage[color=green!30]{todonotes}
\usepackage{xcolor}

\usepackage{graphicx}
\usepackage{amsmath,amssymb,mathtools,bm}
\usepackage{enumitem}
\usepackage{booktabs}
\usepackage{microtype}
\usepackage{array}
\usepackage{derivative}
\usepackage{multirow}
\usepackage{placeins}
\usepackage{float}
\usepackage{tikz}
\usepackage{tabularx} 
\usepackage[colorlinks=true,linkcolor=blue,citecolor=blue,urlcolor=blue]{hyperref}

\usetikzlibrary{shapes.geometric, arrows.meta, positioning, fit, backgrounds}

\DeclareMathOperator*{\argmax}{arg\,max}
\DeclareMathOperator*{\argmin}{arg\,min}

\newcommand{\NSW}{\operatorname{NSW}}

\newcommand{\R}{\mathbb{R}}

\newcommand{\eps}{\varepsilon}

\newcommand{\equalcontrib}{\textsuperscript{$\dagger$}}
\newcommand{\corrauth}{\textsuperscript{*}}
\begin{document}
\title{EF1-Constrained Nash Social Welfare with Identical Additive Valuations: 
Complexity, Guarantees, and Experiments\thanks{The full version of this work with full appendix is available at arXiv:2609.03846.}}
%
%
\author{Zih-Sian Yang\inst{1}\equalcontrib \and
Yi-Hao Chen\inst{1}\equalcontrib \and 
Yu-Te Kuan\inst{1} \and 
Cheng-Jui Wu\inst{1} \and 
Chuang-Chieh Lin\inst{1}\equalcontrib\corrauth \and 
Po-An Chen\inst{2}\corrauth 
}
\authorrunning{T.-H. Yang et al.}
%
\institute{National Taiwan Ocean University, Keelung City 202301, Taiwan\\
\email{\{01257049,01357144,01257109,josephcclin,01057139\}@mail.ntou.edu.tw}\\
 \and
National Yang Ming Chiao Tung University, Hsinchu City 300, Taiwan\\
\email{poanchen@nycu.edu.tw}}
\maketitle              
\begingroup
\renewcommand\thefootnote{}
\footnotetext{\textsuperscript{$\dagger$} The authors contributed equally.}
\footnotetext{\textsuperscript{*} Corresponding authors.}
\endgroup
\begin{abstract}
We study the allocation of indivisible goods among agents with identical additive valuations, focusing on envy-freeness up to one good (EF1) and Nash social welfare (NSW). Since every maximum-NSW allocation is EF1 under additive valuations, the associated threshold problem inherits the known strong \textsf{NP}-hardness of NSW maximization under identical additive valuations and is strongly \textsf{NP}-complete. We therefore focus on welfare guarantees satisfied by arbitrary EF1 allocations. Although every such allocation is known to achieve an $e^{-1/e}$-approximation to the unrestricted optimal NSW, we identify conditions yielding stronger guarantees. Under uniform valuations, every EF1 allocation is NSW-optimal. Under an $\eps$-small-item condition, every EF1 allocation achieves an explicit approximation ratio $\rho_n(\epsilon)$ satisfying $\rho_n(\eps) = 1-O(\eps^2)$ as $\eps\to 0$ for fixed~$n$.

We further consider the stronger sequential requirement that $\operatorname{EF1}$ be maintained after every item assignment. For this setting, we introduce \emph{PriorityNet}, a deep reinforcement learning framework trained with Proximal Policy Optimization (PPO) and equipped with prospective $\operatorname{EF1}$ action masking, which guarantees prefix-wise $\operatorname{EF1}$ by construction. Across 3,000 test instances in each of the offline full-information and random-order online regimes ($n\in[2,20]$, $m\in[5,100]$), PriorityNet achieves mean normalized $\operatorname{NSW}$ values of $0.9911$ and $0.9701$, respectively. Relative to the offline Longest Processing Time (LPT) heuristic and the online least-valued-bundle rule, it attains instance-wise win-minus-loss rates of $+27.10\%$ and $+17.87\%$. Its aggregate welfare matches the offline LPT baseline to four decimal places and modestly improves upon the online baseline, from~$0.9694$ to~$0.9701$.

\end{abstract}
%
%
%
\section{Introduction}
\label{sec:intro}

Fair allocation of indivisible goods is a central problem in artificial intelligence,
computational social choice, and multiagent resource allocation. In this problem,
a set of indivisible goods must be allocated among agents with possibly different
preferences over bundles of goods. Since exact envy-freeness may fail to exist even
in very simple indivisible-goods instances, much of the modern literature focuses
on relaxations of envy-freeness and on welfare objectives that balance efficiency
and fairness.

One particularly influential objective is the \emph{Nash social welfare} (NSW),
defined as the geometric mean of the agents' utilities. Originating from Nash's
bargaining solution~\cite{Nash1950Bargaining}, NSW has become a canonical
objective in fair division because it interpolates between utilitarian welfare and
egalitarian welfare. It rewards efficiency while strongly penalizing allocations in
which some agents receive very low utility. For an allocation
$A=(A_1,A_2,\ldots,A_n)$, the NSW is $\operatorname{NSW}(A) = \big(\prod_{i=1}^n v_i(A_i)\big)^{1/n}$.
When the valuations are additive and the goods are indivisible, maximizing NSW is
computationally challenging. It is \textsf{NP}-hard and even \textsf{APX}-hard in 
general~\cite{NguyenRothe2014Minimizing,Lee2017APXHardness}. Nevertheless, the objective
has received substantial attention due to its strong fairness and efficiency
properties.

In this paper, we study NSW maximization under the fairness constraint of
\emph{envy-freeness up to one good} (EF1). EF1 is one of the most widely used
relaxations of envy-freeness for indivisible goods. An allocation is EF1 if, for
every pair of agents $i$ and $j$, any envy of agent $i$ toward agent $j$ can be
eliminated by removing one good from $j$'s bundle. EF1 allocations are guaranteed
to exist and can be computed efficiently under broad valuation classes
\cite{LiptonEtAl2004ApproxFair}. The relationship between EF1 and NSW is also
well known: maximum Nash welfare allocations are EF1 and Pareto optimal for
additive valuations~\cite{CaragiannisEtAl2019Unreasonable}. Moreover,
Barman, Krishnamurthy, and Vaish showed that, for identical additive valuations,
every EF1 allocation gives an $e^{-1/e}$-approximation to the optimal 
NSW~\cite{BarmanKrishnamurthyVaish2018FairEfficient}.

We focus on the special but important case in which all agents have
\emph{identical additive valuations}. In this setting, every agent assigns the
same value~$v(g)$ to each good~$g$, and the value of a bundle is the sum of the
values of its goods. This restriction captures resource-allocation problems in
which the goods have a common objective quality or utility, and the main issue is
to divide the goods as evenly as possible. In contrast with the case of general
additive valuations, the total utilitarian welfare is fixed across all complete
allocations. Thus, maximizing NSW becomes a balancing problem. That is, one seeks a
partition of the goods whose bundle values are as equal as possible.

Our goal is to understand the interaction between EF1 and NSW in this identical
additive setting. We consider the following EF1-constrained NSW decision problem:
given an identical additive valuation instance and a threshold~$T$, decide whether
there exists an EF1 allocation~$A$ such that $\NSW(A)\geq T$.
We also study structural approximation guarantees: for an arbitrary EF1
allocation $A$, how close must $\operatorname{NSW}(A)$ be to the optimal NSW?
The contributions of this paper are as 
follows.

\subsection{Our Contributions}
\label{subsec:contribution}

\subsubsection{Theoretical Results}

\paragraph{Offline Setting.}
First, we note that EF1-constrained NSW maximization is strongly \textsf{NP}-complete
even under identical additive valuations. The proof, which can be rediscovered in~\cite{CaragiannisEtAl2019Unreasonable,InoueKobayashi2025AdditivePTAS}, 
is by a direct reduction from
\textsc{3-Partition}. The reduction exploits the fact that an allocation with NSW
equal to the average share must have perfectly equal bundle values. Such an
allocation is automatically envy-free, and hence EF1. 
We include a direct proof in the full version for completeness.
Second, we revisit the known $e^{-1/e}$ approximation guarantee for EF1
allocations under identical additive valuations. This guarantee shows that EF1
alone already provides a constant-factor approximation to the unconstrained
optimal NSW. We then identify natural conditions under which this guarantee can
be improved.
Third, as a simple benchmark, we observe that under uniform identical valuations, 
EF1 forces the bundle cardinalities to differ by at most one. 
Consequently, every EF1 allocation is NSW-optimal. 
This provides a simple but useful benchmark:
when all goods have the same value, EF1 forces the bundle cardinalities to be as
balanced as possible, which maximizes the geometric mean.
Finally, we prove an improved approximation guarantee under a small-item condition. Let
$V=v(M),\, \mu=\frac{V}{n},\, \Delta=\max_{g\in M} v(g)$.
If each good has value at most an $\varepsilon$-fraction of the average share,
i.e., $\Delta\leq \varepsilon\mu$, then EF1 implies that all bundle values differ
by at most $\Delta$. This bounded-spread property yields an explicit NSW
approximation ratio that tends to~$1$ as $\varepsilon\to 0$. Thus, in large-market
instances where individual goods are small relative to the average share, EF1
allocations are nearly NSW-optimal. The theoretical results in the offline setting 
are summarized in Table~\ref{tab:theoretical_results}.

\begin{table}[ht]
\footnotesize
\begin{center}
\begin{tabular}{>{\raggedright\arraybackslash}p{0.38\linewidth} >{\raggedright\arraybackslash}p{0.38\linewidth} >{\raggedright\arraybackslash}p{0.25\linewidth}}
\toprule
Setting or problem & Guarantee or Complexity & Status/source \\
\midrule
Identical additive valuations & $e^{-1/e}$-approximation & \cite{BarmanKrishnamurthyVaish2018FairEfficient} \\
\textbf{Uniform identical valuations} & \textbf{optimal (1-approximation)} & \textbf{observation}; cf.~\cite{FreemanEtAl2019Equitable,IgarashiEtAl2025Conflicting} \\
\textbf{Identical valuations with $\bm{\eps}$-small-item condition} & \textbf{$\bm{\rho_n(\eps)}$-approximation, $\bm{\rho_n(\eps)\to 1}$ as $\bm{\eps\to 0}$} & \textbf{this work} \\
Exact EF1-constrained NSW & Strongly \textsf{NP}-complete & Inherited from~\cite{CaragiannisEtAl2019Unreasonable,InoueKobayashi2025AdditivePTAS} \\
\bottomrule
\end{tabular}
\vspace{6pt}
\caption{Offline complexity results and NSW guarantees for EF1 allocations under identical additive valuations.}
\label{tab:theoretical_results}
\end{center}
\end{table}

\paragraph{Online Setting.} 
We first show that the online allocation of indivisible goods does not always have the EF1 guarantee 
even for agents with additive valuations, by providing an illustrating counterexample. Aside from the impossibility result, 
we show that there always exists an EF1 allocation at any time prefix interval when agents have identical additive valuations. 
This further motivates our investigation on NSW maximization under the EF1 constraint for agents with identical additive valuations. 

\subsubsection{Experimental Results}

To translate our theoretical guarantees into scalable algorithms, we propose \emph{PriorityNet}, a deep reinforcement learning framework trained via PPO equipped with prospective $\operatorname{EF1}$ action masking. Under identical additive valuations, fair division shares a natural structural equivalence with multiprocessor scheduling: assigning incoming items to the agent with the lowest accumulated utility mirrors Graham's classical Longest Processing Time (LPT) heuristic~\cite{graham1969bounds}. In the fair division literature, this greedy principle corresponds to the 1.061-NSW approximation algorithm \texttt{Alg\_Identical} in the offline setting~\cite{BarmanKrishnamurthyVaish2018Greedy} and the least-valued-bundle rule in the online streaming setting (Theorem~\ref{thm:least-valued-bundle-ef1}; see also~\cite{ElkindEtAl2025Temporal,NeohPetersTeh2026Online}). With a slight abuse of terminology, we call the LPT heuristic and the least-valued-bundle rule offline and online LPT rules, respectively, for simplicity. We evaluate PriorityNet against these canonical offline and online LPT benchmarks across 3,000 multi-scale instances ($n \in [2, 20], m \in [5, 100]$), demonstrating that PriorityNet achieves near-continuous optimal welfare ($\geq 99.8\%$ and $\geq 97.0\%$ of $\text{MaxNSW}$). PriorityNet also achieves instance-wise win-minus-loss rates of~$+27.10\%$ and $+17.87\%$ with respect to offline and online LPT benchmarks, respectively.

\subsection{Related Work}
\label{subsec:related}

\paragraph{Nash social welfare and indivisible goods.}
The Nash social welfare objective originates from Nash's bargaining solution
\cite{Nash1950Bargaining} and has since become a fundamental welfare criterion in
fair division and multiagent resource allocation. For divisible resources, NSW is
closely connected to market equilibrium and convex programming formulations. For
indivisible goods, however, NSW maximization becomes algorithmically difficult.

Nguyen and Rothe studied the allocation of indivisible goods with additive
utilities and considered, among other objectives, the maximization of average Nash
social welfare~\cite{NguyenRothe2014Minimizing}. They showed that the problem is
computationally challenging and gave approximation algorithms, including a \textsf{PTAS}
for the case of identical additive valuations. Lee later proved that maximizing
NSW with indivisible goods and additive utilities is \textsf{APX}-hard~\cite{Lee2017APXHardness}, 
thereby ruling out a polynomial-time approximation
scheme for the general additive case unless $\mathsf{P}=\mathsf{NP}$.

A major algorithmic breakthrough was obtained by Cole and Gkatzelis, who gave the
first constant-factor approximation algorithm for NSW maximization with
indivisible goods and additive valuations~\cite{ColeGkatzelis2018NSW}. Subsequent
work improved and generalized this line of research using techniques from Fisher
markets, matching, local search, and configuration linear programs. Barman,
Krishnamurthy, and Vaish gave a combinatorial approach that obtains a
$1.45$-approximation and also produces fair and efficient allocations
\cite{BarmanKrishnamurthyVaish2018FairEfficient}. More recently, Feng and Li
obtained an $(e^{1/e}+\varepsilon)$-approximation for the weighted NSW problem
with additive valuations~\cite{FengLi2024WeightedNSW}. Their analysis explicitly
uses the worst-case gap between the optimal NSW and the NSW of an EF1 allocation
in identical additive instances.

\paragraph{Identical additive valuations.}
The identical additive valuation setting is a natural special case in which all
agents agree on the value of each good. Although this setting is more structured
than the general additive case, the problem remains nontrivial: the goal is to
partition the goods into bundles whose total values are as balanced as possible.
Barman, Krishnamurthy, and Vaish studied simple greedy algorithms for special
valuation classes and showed that, for identical additive valuations, a greedy
algorithm that processes goods in decreasing order of value and assigns each good
to a currently least-valued agent gives a $1.061$-approximation to the optimal
NSW~\cite{BarmanKrishnamurthyVaish2018Greedy}. They also observed that NSW
maximization remains \textsf{NP}-hard even under identical valuations. Inoue and Kobayashi
later gave an additive approximation scheme for identical additive valuations,
achieving an additive error of $\varepsilon v_{\max}$, where $v_{\max}$ is the
maximum value of a good~\cite{InoueKobayashi2025AdditivePTAS}.

Our work is complementary to this literature. Prior work on identical additive
valuations primarily studies unconstrained NSW maximization. In contrast, we
focus on the EF1-constrained problem and on approximation guarantees that hold
for every EF1 allocation, not only for the output of a particular algorithm.

\paragraph{EF1 and maximum Nash welfare.}
EF1 was introduced as an algorithmically tractable relaxation of envy-freeness
for indivisible goods~\cite{LiptonEtAl2004ApproxFair}. It has become a standard
fairness notion because EF1 allocations exist under very general conditions,
whereas exact envy-free allocations may not exist. A fundamental connection
between EF1 and NSW was established by Caragiannis et al., who showed that every
maximum Nash welfare allocation for additive valuations is EF1 and Pareto
optimal~\cite{CaragiannisEtAl2019Unreasonable}. This result explains why NSW is
not only an efficiency objective but also a fairness-promoting rule.

Barman, Krishnamurthy, and Vaish strengthened this connection algorithmically by
developing methods for finding allocations that are EF1 and Pareto optimal, and
by proving approximation guarantees for NSW~\cite{BarmanKrishnamurthyVaish2018FairEfficient}.
Of particular relevance to our paper is their result that, under identical
additive valuations, every EF1 allocation is an $e^{-1/e}$-approximation to the
optimal NSW. This guarantee serves as the baseline for our approximation results.
We show that the guarantee can be improved to exact optimality under uniform
identical valuations and can be strengthened under a small-item condition.

\paragraph{Welfare maximization subject to fairness constraints.}
Another related line of work studies the complexity of maximizing welfare
subject to fairness constraints. Aziz, Huang, Mattei, and Segal-Halevi studied
the problem of computing allocations that are both fair and utilitarian-welfare
maximizing, focusing on EF1 and PROP1 constraints~\cite{AzizEtAl2023WelfareMaximizingFair}. 
They showed strong \textsf{NP}-hardness results for several variants. 
Bu et al. studied the complexity and approximability of
maximizing social welfare within EFX and EF1 allocations~\cite{BuEtAl2022ComplexityFairWelfare}. 
These works are conceptually close to
ours because they investigate welfare optimization over fair allocations.
However, they focus primarily on utilitarian social welfare and broader valuation
models. Our paper instead studies the Nash social welfare objective under EF1,
with emphasis on identical additive valuations, where the utilitarian objective
is fixed and the geometric-mean objective captures balance among agents.

\paragraph{Learning-based approaches for fair division.}
Most algorithmic work on NSW maximization and fair division is based on
combinatorial algorithms, market-based methods, local search, or linear-programming
relaxations. Recently, learning-based approaches have also been proposed for fair
allocation of indivisible goods. Mascioli, Goyal, and Chakraborty introduced
\emph{FairFormer}, a transformer architecture for discrete fair division under additive
subjective valuations~\cite{MascioliGoyalChakraborty2026FairFormer}. Their model is an
amortized, permutation-equivariant two-tower transformer that encodes agents and goods
as unordered token sets, applies self-attention within each set, and uses item-to-agent
cross-attention to output assignment distributions. Their method
discretizes the assignment and applies a lightweight \emph{repair} routine to remove
violations of EF1 while attempting to preserve or improve Nash welfare.

This learning-based line of work is closely related to the experimental component of
our paper. Mascioli et al.'s FairFormer~\cite{MascioliGoyalChakraborty2026FairFormer} 
is primarily an amortized neural optimization method for general additive subjective valuations. 
Their approach enforces EF1 through a \emph{post-processing repair step} after discretization, 
whereas our experimental component 
examines whether a transformer-inspired~\cite{vaswani2017attention} model can directly produce allocations with
competitive NSW approximation ratios, without additional repair procedures. Thus,
FairFormer provides an important recent benchmark and motivation for the learning-based
part of our study, while our results complement it by giving problem-specific complexity
and approximation analyses for the identical-valuation EF1-constrained setting.

\paragraph{Positioning of this paper.}
The present paper lies at the intersection of EF1 fairness, Nash social welfare
maximization, and identical additive valuations. On the theoretical side, 
we first observe that the final-time EF1 constraint is nonbinding:
the best EF1 allocation has the same NSW as the unrestricted optimum.
Hence, the constrained decision problem inherits the known strong
\textsf{NP}-completeness of identical-additive NSW maximization.  We include a
direct proof in the full version for completeness. 
Our contribution concerns not the best EF1 allocation, but every EF1 allocation. 
We show that uniform identical valuations imply exact NSW optimality and derive an explicit
small-item approximation ratio $\rho_n(\epsilon)$ satisfying
$\rho_n(\epsilon) = 1-O(\epsilon^2)$. 
On the experimental side, motivated by recent neural approaches
such as FairFormer~\cite{MascioliGoyalChakraborty2026FairFormer}, we investigate
transformer-inspired reinforcement learning methods for producing high-NSW allocations in this structured fair
division setting.

\section{Preliminaries and Problem Specification}
\label{sec:preliminaries}

Let $N=\{1,2,\ldots,n\}$ be a set of agents and let $M$ be a finite set of indivisible goods.  
We assume that each agent~$i\in N$ has an additive valuation function 
$v_i:2^M\to \R_{\geq 0}$, such that for any subset $S\subseteq M$, $v_i(S)=\sum_{g\in S} v_i(g)$. 
When $v_1(g) = v_2(g) = \ldots = v_n(g) := v(g)$ for each good $g\in M$, we say that the agents have \emph{identical} valuation functions.  
An allocation is a partition $A=(A_1,A_2,\ldots,A_n)$ of $M$, where $A_i$ is the bundle assigned to agent $i$. 
The Nash social welfare of an allocation $A$ is the geometric mean of the values of agents' bundles, that is, 
    $\NSW(A)=\left(\prod_{i=1}^n v_i(A_i)\right)^{1/n}$. 
For multiplicative approximation statements we focus on instances for which the optimal NSW is positive. Otherwise, the ratio would be degenerate.

\begin{definition}[EF1]
An allocation $A=(A_1,A_2,\ldots,A_n)$ is envy-free up to one good (EF1) if for every pair of agents $i,j\in N$ with $A_j\ne\emptyset$, there exists a good $g\in A_j$ such that $v_i(A_i)\geq v_i(A_j\setminus\{g\})$.
\end{definition}

The main approximation question considered here is the following.

\medskip
\noindent
\textbf{Approximation question.}
\emph{For an EF1 allocation $A$, how large can we guarantee the ratio $\NSW(A)/\NSW(A^*)$ to be, where~$A^*$ 
is an allocation maximizing NSW among all allocations?}  By the result of Caragiannis et al.~\cite{CaragiannisEtAl2019Unreasonable},
an NSW-maximizing allocation is EF1 whenever the optimal NSW is positive. Therefore, 
$\max_{A\text{ is EF1}} \NSW(A) = \max_A \NSW(A)$.
Thus, requiring EF1 only for the final allocation does not change the
optimal objective value. Our approximation question concerns the
welfare guaranteed by an arbitrary EF1 allocation, rather than the
welfare of the best EF1 allocation.

We specifically consider the following exact decision problem.

\begin{definition}[\textsc{EF1-Identical-NSW}]
The input consists of a set $M$ of goods, a number $n$ of agents, a nonnegative integer value $w_g$ for each good $g\in M$, and an integer threshold $T$.  All agents have the identical additive valuation $v(S)=\sum_{g\in S} w_g$ for any $S\subseteq M$.  
The question is whether there exists an EF1 allocation $A=(A_1,A_2,\ldots,A_n)$ such that $\NSW(A)\geq T$.
Equivalently, since the values are integral, this asks whether there exists an EF1 allocation satisfying $\prod_{i=1}^n v(A_i)\geq T^n$.
\end{definition}

\paragraph{Remark.} By the preceding observation, \textsc{EF1-Identical-NSW} is
equivalent, in terms of its optimal objective value, to ordinary NSW
maximization under identical additive valuations. Consequently, its
computational hardness is inherited from the corresponding
unconstrained problem; the EF1 requirement does not introduce an
additional restriction at an optimum.

\subsection{Known Baseline Guarantee}
\label{subsec:baseline}

A result of Barman, Krishnamurthy, and Vaish~\cite{BarmanKrishnamurthyVaish2018FairEfficient}, as quoted by Feng and Li~\cite{FengLi2024WeightedNSW}, gives the following baseline.

\begin{theorem}[Barman--Krishnamurthy--Vaish~\cite{BarmanKrishnamurthyVaish2018FairEfficient}]
\label{thm:bkv}
For the unweighted Nash social welfare problem with identical additive valuations, every EF1 allocation is an $e^{1/e}$-approximate solution.  Equivalently, if $A$ is EF1 and $A^*$ maximizes NSW, then
$\NSW(A)\geq e^{-1/e}\NSW(A^*)$.
\end{theorem}
Clearly, Theorem~\ref{thm:bkv} implies that EF1 alone already provides a nontrivial constant-factor approximation to optimal NSW, where the approximation ratio is~$e^{-1/e}\approx 0.6922$. In the subsequent sections, we investigate how this guarantee can be improved both theoretically and experimentally.

\section{Theoretical Results: Hardness and Optimization} 
\label{sec:theoretical_results}

\subsection{Equivalence and Known Computational Hardness}
\label{subsec:NPC}

We first note that exact EF1-constrained NSW maximization is computationally hard \emph{even under identical additive valuations}.  The proof is direct and uses \textsc{3-Partition}, which is strongly \textsf{NP}-complete~\cite{gareyjohnson1979}.

\begin{theorem}[Strong \textsf{NP}-completeness]
\label{thm:NPC}
\textsc{EF1-Identical-NSW} is strongly \textsf{NP}-complete.  
\end{theorem}
\begin{proof}[Sketch]
Membership in \textsf{NP} follows because an allocation is a
polynomial-size certificate, and its EF1 property and Nash product
can be verified using polynomial-time integer arithmetic.

For hardness, identical-additive NSW maximization is already known to
be strongly NP-hard via a reduction from \textsc{3-Partition}~\cite{InoueKobayashi2025AdditivePTAS}.  
Moreover, every maximum-NSW allocation is EF1 
under additive valuations~\cite{CaragiannisEtAl2019Unreasonable}.
Therefore, imposing EF1 on the final allocation does not change the
optimal value, and the known strong hardness transfers directly to
\textsc{EF1-Identical-NSW}.
\qed
\end{proof}
For completeness, Appendix~\ref{appendix:proof_theorem_1} provides a direct self-contained reduction from 3-Partition.

\subsection{Special Case I: Uniform Identical Valuations}
\label{subsec:uniform}

We first consider the case in which all goods have the same value.

\begin{definition}[Uniform identical valuation]
The identical additive valuation is uniform if there exists a constant $c>0$ such that $v(g)=c \text{ for every }g\in M$.
Without loss of generality, one may normalize $c=1$.
\end{definition}

\begin{theorem}[Exact optimality under uniform valuations]
Assume uniform identical valuations.  Then every EF1 allocation is NSW-optimal.  Hence every EF1 allocation is a $1$-approximation to the optimal NSW.
\end{theorem}

\begin{proof}
Normalize $v(g)=1$ for every good $g$.  For an allocation $A$, let us write $x_i=|A_i|=v(A_i)$. 
The NSW objective is $\NSW(A)=(\prod_{i=1}^n x_i)^{1/n}$. 
If $A$ is EF1, then for every pair $i,j$ with $A_j\ne\emptyset$, there exists $g\in A_j$ such that $|A_i|\geq |A_j\setminus\{g\}|=|A_j|-1$.
Therefore $|A_j|-|A_i|\leq 1$ whenever $A_j\neq\emptyset$. Note that if $A_j=\emptyset$, the inequality is immediate.  Hence all bundle sizes differ by at most one.

Now consider any integer vector $(x_1,x_2,\ldots,x_n)$ with $\sum_{i=1}^n x_i=m$.  If there exist $a,b\in\{1,2,\ldots,n\}$ with $x_a\ge x_b+2$, then replacing $(x_a,x_b)$ by $(x_a-1,x_b+1)$ strictly increases the product. Indeed, $(x_a-1)(x_b+1)-x_ax_b = x_a-x_b-1 > 0$.
Thus the product $\prod_{i=1}^n x_i$ is maximized exactly when all $x_i$ differ by at most one.  Since every EF1 allocation has this property, every EF1 allocation maximizes the product and therefore maximizes NSW.
\qed
\end{proof}

\begin{remark}
This result is stronger than the $e^{-1/e}$ baseline, but it relies on a restrictive valuation class.  For completely homogeneous goods, EF1 already enforces the bundle cardinalities that maximize NSW. 
\end{remark}

\subsection{Special Case II: Small-Item Condition}
\label{subsec:small-item}

The uniform case is very special.  A more flexible condition is that no single good is too large compared with the average share.  Define
$V=v(M), \, \mu=\frac{V}{n}, \, v^*=\max_{g\in M}v(g)$.
Here $\mu$ is the average total value per agent, and $v^*$ is the largest item value.

\begin{definition}[Small-item condition]
For $\eps\in[0,1]$, an identical additive instance satisfies the $\eps$-small-item condition if
$v^*\leq \eps \mu$. Equivalently, every individual good is worth at most an $\eps$-fraction of~$mu$. 
\end{definition}
This condition is natural in large-market settings: as goods become individually negligible relative to the average share, EF1 forces the realized bundle values to become nearly balanced. For $\eps\in[0,1]$, let us define
\[
    \rho_n(\eps)
    =
    \left[
    \min_{1\leq k\leq n-1}
    \left(1-\frac{k\eps}{n}\right)^{n-k}
    \left(1+\frac{(n-k)\eps}{n}\right)^k
    \right]^{1/n}.
\]

\begin{theorem}[Improved NSW guarantee under small items]
\label{thm:improved_nsw_small_item}
Assume identical additive valuations and the $\eps$-small-item condition $v^*\leq\eps\mu$ for $\eps\in[0,1]$.  Then every EF1 allocation $A$ satisfies $\NSW(A)/\NSW(A^*) \geq \max\{e^{-1/e},\rho_n(\eps)\}$, where $A^*$ is an NSW-optimal allocation.  
\end{theorem}
\begin{proof}
Refer to Appendix~\ref{appendix:proof_theorem_4} for the proof in the full paper. 
\qed
\end{proof}

\begin{remark}[Asymptotics]
For fixed $n$, we can show that $\rho_n(\eps)\to 1$ as $\eps\to 0$.  Expanding the logarithm of the $k$th term by Taylor expansion gives
\[
    \frac{1}{n}\log\left[
    \left(1-\frac{k\eps}{n}\right)^{n-k}
    \left(1+\frac{(n-k)\eps}{n}\right)^k
    \right]
    =
    -\frac{k(n-k)}{2n^2}\eps^2+O(\eps^3).
\]
The worst second-order term is obtained by maximizing $k(n-k)$, so $\rho_n(\eps) = 1-\frac{\lfloor n^2/4\rfloor}{2n^2}\eps^2+O(\eps^3)\geq 1-\frac{1}{8}\eps^2 - O(\eps^3) \to 1\text{ as } \eps\to 0$.
\end{remark}

\subsection{Positive and Negative Results on the Online EF1 Allocation}

Consider a set $N$ of agents and a sequence of indivisible goods
$g_1,g_2,\ldots$ arriving online. When $g_t$ arrives, its values
$(v_i(g_t))_{i\in N}$ are revealed and the good must be allocated
immediately and irrevocably. Let
$A^t=(A_1^t,A_2^t,\ldots,A_n^t)$ denote the accumulated allocation after the
first $t$ arrivals. Recall that $A^t$ is \emph{envy-free up to one good}
(EF1) if, for every pair $i,j\in N$ with $A_j^t\neq\varnothing$, there
exists $g\in A_j^t$ such that $v_i(A_i^t)\geq v_i(A_j^t\setminus\{g\})$.
An online allocation sequence $(A^t)_{t=1}^T$ is \emph{prefix-wise EF1} 
if $A^t$ is EF1 for every $t\in \{1,2,\ldots, T\}$.

The impossibility of maintaining exact fairness under irrevocable online
allocation is well established. Benad{\`e} et al.~\cite[Theorem~2.13]{BenadeEtAl2018EnvyVanish}
show that, even when every single-good value lies in $[0,1]$, an
adversary can force the maximum pairwise envy to grow polynomially with
the horizon. Since an EF1 allocation with single-good values bounded by
$1$ has pairwise envy at most $1$, their lower bound already rules out
a general EF1 guarantee. More directly, He et
al.~\cite[Theorem~3.3]{HeEtAl2019FairerFuture} study the requirement
that EF1 hold after every round and prove that any \emph{uninformed}
online EF1 algorithm requires at least $\lfloor T/6\rfloor$
reallocations in the worst case, even with two agents. In particular,
an irrevocable algorithm, which permits zero reallocations, cannot
guarantee prefix-wise EF1. We give below a short self-contained
counterexample tailored to deterministic irrevocable algorithms.

\begin{proposition}[Impossibility of prefix-wise EF1 for general additive valuations]
\label{prop:online-ef1-impossibility}
For general nonnegative additive valuations, no deterministic
irrevocable online allocation algorithm can guarantee that the
accumulated allocation $A^t$ is EF1 after every arriving good $g_t$.
\end{proposition}

\begin{proof}
It suffices to consider two agents. Fix any constant $K>1$, and let
$\mathcal A$ be an arbitrary deterministic irrevocable online algorithm.
We construct an adversarial input for $\mathcal A$.

The first arriving good $g_1$ has $v_1(g_1)=v_2(g_1)=1$.
Since $\mathcal A$ is deterministic, it assigns $g_1$ to one of the two
agents. WLOG, rename the recipient agent~$1$ and
the other agent~$2$. Hence, $A_1^1=\{g_1\},\, A_2^1=\varnothing$.

The second good $g_2$ has values $v_1(g_2)=K,\, v_2(g_2)=\frac1K$.
If $\mathcal A$ assigns~$g_2$ to agent~$1$, then
$A_1^2=\{g_1,g_2\}$ and $A_2^2=\varnothing$. From agent~$2$'s
perspective, deleting either one good from $A_1^2$ leaves strictly
positive value, that is, $v_2(A_1^2\setminus\{g_1\})=\frac1K>0=v_2(A_2^2)$ and 
$v_2(A_1^2\setminus\{g_2\})=1>0=v_2(A_2^2)$.
Thus $A^2$ would not be EF1. Consequently, any algorithm that preserves
EF1 after the second arrival is forced to assign $g_2$ to agent~$2$, so
$A_1^2=\{g_1\},\, A_2^2=\{g_2\}$.

Now let a third good $g_3$ arrive with $v_1(g_3)=v_2(g_3)=K$.
There are only two possible irrevocable assignments. If~$g_3$ is given
to agent~$1$, then agent~$2$ has value~$1/K$, whereas deleting either
good from agent~$1$'s bundle leaves value strictly larger than $1/K$, that is, 
$v_2(\{g_3\})=K>\frac1K, \, v_2(\{g_1\})=1>\frac1K$.
Hence EF1 fails for the ordered pair~$(2,1)$. If instead~$g_3$ is given
to agent~$2$, then agent~$1$ has value~$1$, while deleting either good
from $A_2^3=\{g_2,g_3\}$ leaves value $K>1$ according to agent~$1$. That is, 
$v_1(\{g_2\})=K>1, \, v_1(\{g_3\})=K>1$.
Hence EF1 fails for the ordered pair~$(1,2)$. Therefore no assignment of
$g_3$ preserves EF1, contradicting the claimed guarantee of
$\mathcal A$. Since $\mathcal A$ was arbitrary, no deterministic
irrevocable online algorithm can guarantee prefix-wise EF1 for the full
class of additive valuations.
\qed
\end{proof}

For identical valuations the situation changes. Suppose all agents
share the same nonnegative additive valuation~$v$. Consider the
\emph{least-valued-bundle rule}: when $g_t$ arrives, we choose $r\in\argmin_{i\in N} v(A_i^{t-1})$
and assign~$g_t$ to agent~$r$.
The following result is also covered by a stronger result of Elkind et
al.~\cite[Theorem~3.7]{ElkindEtAl2025Temporal}. They prove temporal EF1
(TEF1), i.e., EF1 at every prefix, for the more general class of
\emph{generalized binary} (restricted additive) valuations. Their
greedy algorithm for goods assigns each positively valued arriving
good to an interested agent whose current bundle has minimum value
(see Algorithm~3 in their full version). Under identical nonnegative
valuations, every positive-valued good is valued by every agent, so
their rule specializes exactly to the least-valued-bundle rule below.

\begin{theorem}[Prefix-wise EF1 under identical additive valuations]
\label{thm:least-valued-bundle-ef1}
If all agents have the same nonnegative additive valuation $v$, then the
least-valued-bundle rule maintains EF1 after every arriving good.
\end{theorem}

\begin{proof}
We proceed by induction on~$t$. All bundles are empty initially, so EF1
holds. Assume that $A^{t-1}$ is EF1, and let
$r\in\argmin_{i\in N}v(A_i^{t-1})$ receive the newly arriving good
$g_t$. Thus $A_r^t=A_r^{t-1}\cup\{g_t\}, \,A_i^t=A_i^{t-1}\,(i\neq r)$.
Note that only agent $r$'s bundle changes. Hence, for every ordered pair~$(i,j)$
with~$j\neq r$, the EF1 witness from time~$t-1$ remains valid. Indeed, if
$i\neq r$, neither relevant bundle changes, while if $i=r$, the value of
$i$'s own bundle can only increase.

It remains to consider possible new envy toward agent~$r$. Fix any
$i\neq r$. Since $r$ had a least-valued bundle immediately before the
arrival of~$g_t$, we obtain that $v(A_r^{t-1})\leq v(A_i^{t-1})$.
Agent~$i$'s bundle is unchanged, and removing the newly assigned good
from agent $r$'s bundle restores the old bundle. Therefore,  
$v(A_i^t) = v(A_i^{t-1}) \geq v(A_r^{t-1}) = v(A_r^t\setminus\{g_t\})$.
Thus $g_t$ itself is an EF1 witness for every possible new envy toward
$r$. Consequently $A^t$ is EF1. By induction, the allocation is EF1
after every arrival.
\qed
\end{proof}

\begin{remark}
In an online Nash-social-welfare formulation with a prefix-wise EF1
constraint, Proposition~\ref{prop:online-ef1-impossibility} shows that
feasibility already fails on the unrestricted additive domain.
Theorem~\ref{thm:least-valued-bundle-ef1} identifies the identical-value
domain as a natural setting in which prefix-wise EF1 feasibility is
restored, leaving welfare maximization as a feasible algorithmic
objective.
\end{remark}

\paragraph{Final-time EF1 versus prefix-wise EF1.}
The requirement of maintaining $\operatorname{EF1}$ after every assignment can be strictly stronger than requiring $\operatorname{EF1}$ only for the final allocation. To illustrate this distinction, consider two agents with identical additive valuations and five goods processed in descending order, with values $3,3,2,2,2$. The final allocation $(A_1,A_2)=(\{3,3\},\{2,2,2\})$ gives utility profile $(6,6)$ and therefore attains the optimal Nash social welfare $\operatorname{NSW}(A)=6$. 
In particular, it is envy-free and hence $\operatorname{EF1}$. However, reaching this allocation requires assigning the first two value-$3$ goods to the same agent. Immediately afterward, the partial allocation would be $(\{3,3\},\varnothing)$, which violates $\operatorname{EF1}$ because removing either value-$3$ good from the first bundle still leaves value $3>0$ for the empty-bundle agent. Consequently, any policy maintaining prefix-wise $\operatorname{EF1}$ under this fixed order must assign the first two goods to different agents. Each agent then holds one value-$3$ good, and the most balanced possible distribution of the three remaining value-$2$ goods yields utilities $(7,5)$, with $\operatorname{NSW}=\sqrt{7\cdot5}=\sqrt{35}<6$.
Thus, prospective $\operatorname{EF1}$ action masking may exclude a final-time $\operatorname{EF1}$-optimal allocation, showing that prefix-wise $\operatorname{EF1}$ maximization and final-time $\operatorname{EF1}$ maximization are generally different optimization problems.

\section{Reinforcement Learning Optimization via PPO PriorityNet}
\label{sec:transformer}

Having studied welfare guarantees for final EF1 allocations, we now consider a different and stronger sequential problem. 
Given a fixed item order, the allocation must remain EF1 after every assignment. We introduce \emph{PriorityNet}, 
which is a reinforcement learning architecture trained via Proximal Policy Optimization (PPO)~\cite{schulman2017proximal}, 
for solving this prefix-wise EF1 problem. Note that it does not in general optimize over all final-time EF1 allocations.
PriorityNet parameterizes dynamic priority fields over agents to guide sequential allocations toward maximum welfare. To ensure strict fairness, we integrate a \emph{prospective EF1 action mask} directly into the policy: at every decision step, actions that would violate the EF1 invariant are strictly filtered, guaranteeing exactly prefix-wise EF1 by construction without heuristic post-processing. This section presents the formal MDP formulation, the modular neural architecture, the PPO policy optimization pipeline, and extensive empirical benchmarks across offline and online streaming regimes.

\subsection{Problem Formulation \& Allocation Protocol}
\label{subsec:rbs_formulation}

\subsubsection{Constrained Sequential Allocation as an MDP.}
We formulate fair sequential allocation as a finite-horizon Constrained Markov Decision Process $(\mathcal{S}, \mathcal{A}, \mathbf{M}, \mathcal{P}, \mathcal{R})$. At step $t \in \{1,2, \ldots, m\}$, exactly one good $g_t$ is allocated:
\begin{itemize}[leftmargin=*]
    \item \emph{State} $s_t \in \mathcal{S}$: Encodes the current partial allocation $(A_1^{t-1},A_2^{t-1}, \ldots, A_n^{t-1})$, agent utility state features $\mathbf{F}_t \in [0, 1]^{n \times 11}$, and visible item valuations.
    \item \emph{Action Space \& EF1 Mask} $\mathcal{A}=\{1, 2, \ldots, n\}$: The recipient agent $a_t \in \mathcal{A}$. The prospective EF1 mask $\mathbf{M}_t \in \{0, 1\}^n$ restricts decisions to the feasible set $A_t^{\mathrm{feas}} = \{i \in \mathcal{A} \mid A^{t-1} \cup \{(i, g_t)\} \text{ is } \operatorname{EF1}\}$.
    \item \emph{Policy Decision}: The Actor network outputs continuous priority logits $\boldsymbol{\ell}_t \in \mathbb{R}^n$. The allocated recipient is selected via $a_t = \argmax_{i \in A_t^{\mathrm{feas}}} \ell_{t, i}$ during deterministic evaluation.
    \item \emph{State Transition} $\mathcal{P}$: The chosen agent receives $g_t$ ($A_{a_t}^t = A_{a_t}^{t-1} \cup \{g_t\}$), and the environment advances to~$g_{t+1}$.
\end{itemize}

\subsubsection{Offline vs. Online Allocation Regimes.}
We evaluate PriorityNet under two canonical information regimes sharing the identical sequential masked allocation mechanism:
\begin{enumerate}[leftmargin=*,label=(\alph*)]
    \item \emph{Offline Full-Information Regime}: The complete goods catalog and valuations are known \emph{a priori}. Goods are pre-sorted in descending value order ($v(g_1) \geq v(g_2)\geq \dots \ge v(g_m)$), and unallocated goods are dynamically tracked via global attention.
    \item \emph{Online Dynamic Streaming Regime}: Goods arrive sequentially under the \emph{Random-Order Arrival Model} (window $W=1$) with zero future lookahead. Each arriving good must be irrevocably allocated upon arrival.
\end{enumerate}

\subsubsection{Nash Social Welfare Objective and Continuous Optimal Ceiling.}
For $n$ agents with final accumulated utilities $\mathbf{u} = (u_1, u_2, \dots, u_n)$, we obtain 
$\text{NSW}(\mathbf{u}) = (\prod_{i=1}^{n} u_i)^{1/n} = \exp (\frac{1}{n} \sum_{i=1}^{n} \ln u_i)$. To ensure scale-invariant rewards across varying market dimensions, utility geometric means are normalized by the continuous AM--GM upper bound ceiling (referred to as the continuous optimal ceiling $\text{MaxNSW}$). That is, 
    $\text{MaxNSW} = \frac{1}{n} \sum_{g=1}^{m} \max_{i=1,2,\ldots,n} v_{i,g}$, and define $R_{\text{norm}} = \frac{\text{NSW}(\mathbf{u})}{\text{MaxNSW}}$ as the normalized NSW.
Crucially, this full-instance ceiling is computed strictly post-hoc upon
allocation termination (i.e., $t = m$) for external benchmarking and terminal reward normalization; it is never
accessible to the online policy during sequential decision-making.

\subsubsection{Valuation Generators and Curriculum Training Setup.}
Each problem instance generates a valuation matrix $\mathbf{V} \in \{1,2, \ldots, 100\}^{n \times m}$. Under identical additive valuations ($v_{i,g} = v_{1,g}$ for all $i$), item values are drawn from five representative generator classes: (1) \textit{Uniform}, (2) \textit{Gaussian} $N(\mu, \sigma^2)$, (3) \textit{Bimodal Beta Mixture} (high variance), (4) \textit{Heavy-Tailed Exponential / Skewed} (dominant star items), and (5) \textit{Flat} (low variance). During training, market dimensions $(n, m)$ are sampled dynamically across $n \in [2, 20]$ and $m \in [3, 40]$ which subject to $m \geq n$. Batches consist of $B = 64$ parallel environment workers running up to a trajectory horizon length of $T = 32$ steps. Hyperparameters are detailed in Appendix Table~\ref{tab:ppo_hyperparameters}.

\subsection{State Representation \& PriorityNet Architecture}
\label{subsec:architecture}

At each decision step $t$, the environment constructs an 11-dimensional normalized state feature matrix $\mathbf{F} \in [0, 1]^{n \times 11}$ for the active agents (formally defined in Table~\ref{tab:state_features} in the Appendix). In the offline setting, the full sequence of unallocated goods is visible and processed in descending value order. In the online sequential setting, goods arrive one at a time, so only the currently arriving good $g_t$ is visible ($\mathcal{W}_t=\{g_t\}$) while future goods remain unrevealed. All item valuations are normalized by $100.0$ (matching the valuation support $v \in [1, 100]$), avoiding future-value leakage. Agent cumulative utilities and envy gaps are normalized relative to $\max(U_t^{\max}, 1.0)$ where $U_t^{\max} = \max_{i \in N} u_i$. Crucially, the running normalized welfare is computed causally as $\widehat{\mathrm{NSW}} = \mathrm{NSW}_{\mathrm{cur}} / \max(\mathrm{MaxNSW}_t, 10^{-8})$, where $\mathrm{MaxNSW}_t = \frac{1}{n} \sum_{g \in \mathcal{H}_t} v(g)$ is evaluated strictly over the revealed history $\mathcal{H}_t = \{g_1, g_2,\dots, g_t\}$ ($\mathcal{H}_t = M$ in the offline setting), ensuring a strict zero-lookahead online policy. Agent utility ranks, utility shares, One-Hot last-selected indicators $\mathbb{I}(i = a_{t-1})$, and idle starvation counters are tracked dynamically across steps.

As illustrated in Fig.~\ref{fig:prioritynet}, \textit{PriorityNet} processes agent state features and item valuations through a decoupled dual-stream encoder followed by stacked multi-head attention blocks:
\begin{enumerate}[leftmargin=*,label=(\arabic*)]
    \item \emph{Goods Stream (Item Encoder)}: In the offline setting, a two-layer 1D-CNN backbone ($\text{Conv1D}(1 \to 128 \to 64, k=3, \text{padding}=\text{SAME})$ with ReLU activations) operates along the goods sequence dimension on the identical valuation vector, followed by a linear projection ($\text{Dense}(64 \to 128)$). This extracts unpooled item token embeddings $\mathbf{K}_{\text{items}}, \mathbf{V}_{\text{items}} \in \mathbb{R}^{B \times m \times 128}$ that preserve individual item identities for dynamic attention. In the online single-item streaming setting ($W=1$), the arriving scalar valuation $v_t/100$ is directly projected via an item dense layer to $[B, 1, 128]$.
    \item \emph{Agent Stream (State MLP Projection)}: The 11-dimensional agent state feature matrix $\mathbf{F} \in [0, 1]^{B \times n \times 11}$ is independently projected via a multi-layer perceptron ($\text{Dense}(11 \to 128) \to \text{LayerNorm} \to \text{ReLU} \to \text{Dense}(128 \to 128)$) to produce agent Query representations $\mathbf{Q}_{\text{agents}} \in \mathbb{R}^{B \times n \times 128}$.
    \item \emph{Cross-Attention Block}: A Pre-LayerNorm Multi-Head Cross-Attention block (4 heads, hidden dimension $d=128$)~\cite{vaswani2017attention} enables agent Queries $\mathbf{Q}_{\text{agents}}$ to attend dynamically to item Keys and Values $\mathbf{K}_{\text{items}}, \mathbf{V}_{\text{items}}$ under dynamic attention masking ($\text{valid\_queries} \times \text{valid\_keys}$), followed by a residual connection, LayerNorm, and a deep feedforward network ($\text{DeepFFN}(128 \to 512 \to 128)$).
    \item \emph{Agent Self-Attention Block}: A Pre-LayerNorm Multi-Head Self-Attention block (4 heads, hidden dimension $d=128$) over the $n$ agents computes game-theoretic envy dependencies and relative priority trade-offs, followed by a residual connection, LayerNorm, and DeepFFN.
    \item \emph{Dual Output Heads}: The Actor head projects the encoded agent tokens via a single linear layer ($\text{Dense}(128 \to 1)$) to continuous priority logits $\boldsymbol{\ell}_t \in \mathbb{R}^n$, which are subsequently filtered by the prospective EF1 action mask $\mathbf{M}_t$. The Critic head aggregates agent tokens via masked mean pooling into a global state vector in $\mathbb{R}^{128}$ and evaluates expected social welfare via a value MLP ($\text{Dense}(128 \to 128) \to \text{ReLU} \to \text{Dense}(128 \to 1)$) to estimate state value~$V(s) \in \mathbb{R}$.
\end{enumerate}

\subsection{Policy Optimization via Proximal Policy Optimization}
\label{subsec:ppo_optimization}

\subsubsection{Rationale for Proximal Policy Optimization.}
We use an EF1 action mask to ensure that every selected allocation \emph{remains EF1 throughout the allocation process}. The mask restricts the policy to prospectively EF1-safe agents, reducing the feasible action space to fairness-preserving choices while still leaving multiple valid decisions at many states. PPO is then used to optimize the policy within this constrained action space. Its clipped surrogate objective stabilizes policy updates, preventing outlier high-reward trajectories from inducing destructively large policy shifts. This allows the model to systematically discover non-greedy, welfare-maximizing sequential allocations while preserving the EF1 feasibility guaranteed by action masking.

\subsubsection{Masked Categorical Policy \& Action Sampling.}
\label{subsubsec}

At each allocation step $t \in \{1,2,\ldots, m\}$, the Actor produces priority
logits $\boldsymbol{\ell}_t \in \mathbb{R}^n$. Before sampling an action, we
apply a prospective EF1 mask that removes any agent whose selection would
violate EF1 after assigning the current good. The policy therefore operates
only over the feasible action set $A_t^{\mathrm{feas}}$.

Under identical nonnegative additive valuations,
Theorem~\ref{thm:least-valued-bundle-ef1} guarantees that assigning the
arriving good to an agent with a least-valued current bundle preserves EF1 at
every prefix. Hence, $A_t^{\mathrm{feas}} \neq \emptyset$ for every allocation step $t$,
so the masked policy always admits at least one valid action and no fallback
allocation rule is required.

During training, the recipient is sampled from the resulting masked
categorical policy, enabling stochastic exploration exclusively among
EF1-feasible allocations. During deterministic evaluation and inference, the
highest-priority feasible agent is selected.
The complete masked policy definition, action-selection rules, and masked PPO
likelihood-ratio formulation are provided in
Appendix~\ref{app:masked_policy}.

\subsubsection{Relative Advantage Reward Formulation.}
To direct policy optimization toward welfare dominance over canonical greedy heuristics while standardizing reward scales across instance sizes, we define a normalized relative advantage terminal reward:
\[
    R_{\text{term}} = w_{\text{reward}} \cdot \frac{\text{NSW}_{\text{RL}} - \text{NSW}_{\text{LPT}}}{\text{MaxNSW}},
\]
where $\text{NSW}_{\text{LPT}}$ denotes the welfare achieved by Offline LPT~\cite{BarmanKrishnamurthyVaish2018Greedy} in the offline setting, and Online LPT (Theorem~\ref{thm:least-valued-bundle-ef1}) in the online streaming setting. Here, $\text{MaxNSW} = \frac{1}{n} \sum_{k=1}^m v_k$ is the theoretical welfare ceiling, and $w_{\text{reward}} = 20.0$ scales minute welfare differentials into $[-2.0, +2.0]$. Intermediate step rewards are zero ($r_t = 0$ for $t < m$).

\subsubsection{Surrogate Loss Objective.}

The network parameters $\theta$ are trained end-to-end using the standard PPO
clipped surrogate objective~\cite{schulman2017proximal}, together with
value-function regression and policy entropy regularization. 
Temporal advantages are estimated using Generalized Advantage Estimation
(GAE-$\lambda$)~\cite{schulman2015high} and standardized within each
mini-batch before policy optimization. The overall training objective is
\[
\mathcal{L}_{\mathrm{total}}(\theta)
=
\mathcal{L}_{\mathrm{clip}}(\theta)
+
c_{vf}\mathcal{L}_{\mathrm{value}}(\theta)
-
c_{\mathrm{ent}}\mathcal{H}(\pi_\theta),
\]
where $c_{vf}$ and $c_{\mathrm{ent}}$ control the Critic loss and entropy
regularization, respectively. The PPO likelihood ratio is evaluated under
the same state-dependent EF1 mask used during rollout collection.

The complete GAE formulation, Critic target, advantage normalization,
clipped surrogate objective, PPO likelihood-ratio computation and optimization details are
provided in Appendix~\ref{app:ppo_objective}.

\begin{figure}[!ht]
\centering
\vspace{-4pt}
\resizebox{0.70\textwidth}{!}{%
\begin{tikzpicture}[
    node distance=0.7cm and 0.6cm,
    font=\small\sffamily,
    goodsbox/.style={draw=green!60!black, fill=green!5, thick, rectangle, rounded corners=3pt, align=center, inner sep=5pt, text width=4.8cm},
    agentbox/.style={draw=blue!80!black, fill=blue!5, thick, rectangle, rounded corners=3pt, align=center, inner sep=5pt, text width=4.8cm},
    transbox/.style={draw=orange!80!black, fill=orange!5, thick, rectangle, rounded corners=3pt, align=center, inner sep=5pt, text width=10.6cm},
    headbox/.style={draw=purple!80!black, fill=purple!5, thick, rectangle, rounded corners=3pt, align=center, inner sep=5pt, text width=5.0cm, minimum height=3.0cm},
    groupbox/.style={draw=#1, dashed, thick, rectangle, rounded corners=6pt, inner sep=8pt},
    arrow/.style={-Stealth, thick, draw=gray!80!black}
]

\node [goodsbox, minimum height=1.3cm] (val) at (-2.9, 0) {Item Valuations $\mathbf{V}/100$\\$[B, n, m] \implies \text{Row } 0: [B, m, 1]$};
\node [agentbox, minimum height=1.3cm] (feat) at (2.9, 0) {Agent State Features $\mathbf{F}$\\$[B, n, 11]$ (Table~\ref{tab:state_features})};

\node [goodsbox, minimum height=2.55cm, below=0.7cm of val] (cnn) {1D-CNN Goods Backbone\\$\text{Conv1D}(1 \to 128 \to 64, k=3, \text{SAME})$\\$\text{Item Dense}(64 \to 128) \implies [B, m, 128]$};
\node [agentbox, minimum height=2.55cm, below=0.7cm of feat] (aproj) {Agent State MLP Projection\\$\text{Dense}(11 \to 128) \to \text{LN} \to \text{ReLU} \to \text{Dense}(128)$\\$\implies \text{Agent Queries } [B, n, 128]$};

\node [transbox, below=1.1cm of cnn, xshift=2.9cm] (ca) {Cross-Attention Block ($\times 1$)\\Pre-Norm Multi-Head Cross-Attention (4 heads, $d=128$)\\$\text{Queries: Agents } [B, n, 128], \quad \text{Keys/Values: Items } [B, m, 128]$\\$\text{Residual} + \text{LN} + \text{DeepFFN}(128 \to 512 \to 128) \implies [B, n, 128]$};

\node [transbox, below=0.7cm of ca] (sa) {Agent Self-Attention Block ($\times 1$)\\Pre-Norm Multi-Head Self-Attention (4 heads, $d=128$)\\Agent-to-Agent Mutual Game Interaction\\$\text{Residual} + \text{LN} + \text{DeepFFN}(128 \to 512 \to 128) \implies [B, n, 128]$};

\node [headbox, below left=1.1cm and -2.5cm of sa] (actor) {Actor Head (Policy)\\$\text{Dense}(128 \to 1) \implies \text{Priority Logits } \boldsymbol{\ell}_t \in \mathbb{R}^n$\\[3pt]$\Downarrow$\\[2pt]$\text{Prospective EF1 Action Mask } \mathbf{M}_t$};
\node [headbox, below right=1.1cm and -2.5cm of sa] (critic) {Critic Head (Value)\\Masked Mean Pool over Agents $\to [B, 128]$\\$\text{Dense}(128 \to 128) \to \text{ReLU} \to \text{Dense}(1)$\\$\implies \text{State Value } V(s) \in \mathbb{R}$};

\draw [arrow] (val) -- (cnn);
\draw [arrow] (feat) -- (aproj);
\draw [arrow] (cnn.south) -- node [left, font=\footnotesize\sffamily, text=green!50!black] {$\mathbf{K}_{\text{items}}, \mathbf{V}_{\text{items}}$} ([xshift=-2.9cm]ca.north);
\draw [arrow] (aproj.south) -- node [right, font=\footnotesize\sffamily, text=blue!70!black] {Queries $\mathbf{Q}_{\text{agents}}$} ([xshift=2.9cm]ca.north);
\draw [arrow] (ca) -- (sa);
\draw [arrow] ([xshift=-2.5cm]sa.south) -- (actor.north);
\draw [arrow] ([xshift=2.5cm]sa.south) -- (critic.north);

\begin{scope}[on background layer]
    \node [groupbox=green!60!black, fit=(val) (cnn), label={[green!60!black, font=\bfseries\small, anchor=south west]north west:Goods Stream (Keys/Values)}] {};
    \node [groupbox=blue!80!black, fit=(feat) (aproj), label={[blue!80!black, font=\bfseries\small, anchor=south east]north east:Agent Stream (Queries)}] {};
    \node [groupbox=orange!80!black, fit=(ca) (sa), label={[orange!80!black, font=\bfseries\small, align=center]above:PriorityNet Core Transformer}] {};
    \node [groupbox=purple!80!black, fit=(actor) (critic), label={[purple!80!black, font=\bfseries\small, anchor=south west]north west:Dual Output Heads}] {};
\end{scope}

\end{tikzpicture}%
}
\vspace{-6pt}
\caption{Data flow and modular architecture of \textit{PriorityNet}. Agent state features and item valuations are processed along decoupled streams into Queries and Keys/Values, respectively, and dynamically matched through stacked Cross-Attention and Self-Attention layers before feeding the Actor and Critic heads.}
\label{fig:prioritynet}
\end{figure}
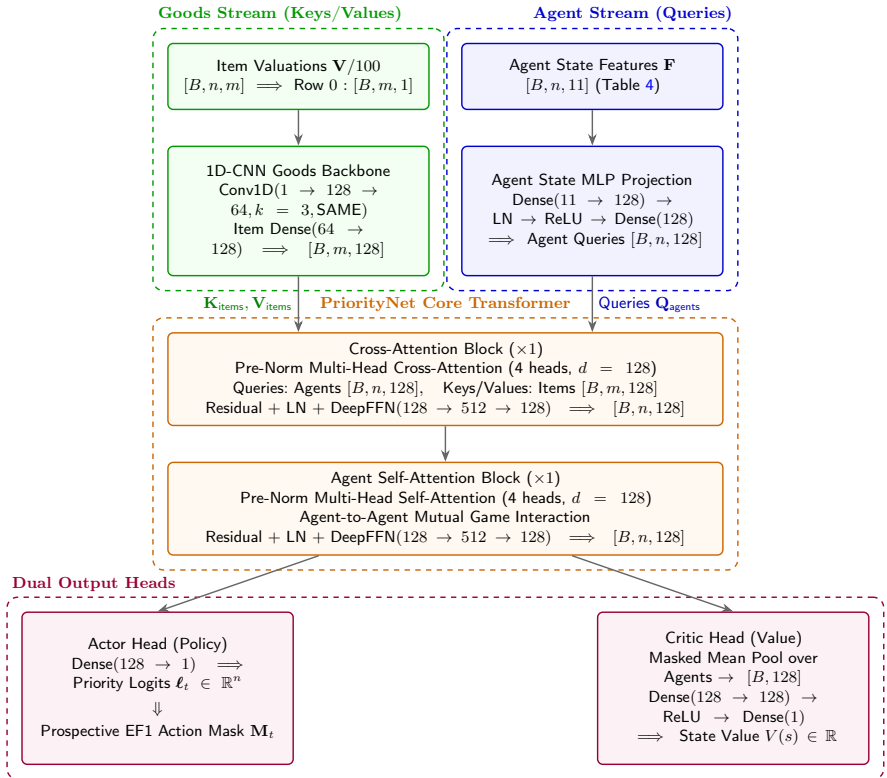

\FloatBarrier

\subsection{Experimental Results \& Comparative Analysis}
\label{subsec:experiments}

\subsubsection{Offline Full-Information Benchmark.}
We evaluate policy performance under the static full-information offline setting across 3,000 independent test instances (1,000 per scale regime: Small $n \in [2, 5], m \in [5, 20]$; Medium $n \in [6, 10], m \in [20, 50]$; and Large $n \in [11, 20], m \in [50, 100]$) drawn from the five representative valuation generator classes under identical additive valuations. In this setting, the full catalog of goods is sorted in descending order and allocated sequentially. We benchmark \textit{PriorityNet} (PPO, converged checkpoint $\text{PriorityNet}_{8\text{k}}$) against four baselines:
\begin{enumerate}[leftmargin=*,label=(\alph*)]
    \item \emph{Offline LPT (Greedy)}: Drawing inspiration from Graham's classical Longest Processing Time first rule in multiprocessor scheduling~\cite{graham1969bounds}, this greedy heuristic pre-sorts goods in descending order and iteratively assigns each item to the agent with the lowest accumulated utility ($\argmin_i u_i$). In fair division literature, this exact protocol was analyzed by Barman et al.~\cite{BarmanKrishnamurthyVaish2018Greedy} as \texttt{Alg\_Identical}, which guarantees EFX and a 1.061-approximation to optimal NSW under identical valuations.
    \item \emph{Reverse Round-Robin (Reverse-RR)}: The classical non-adaptive fair division protocol adopting alternating snake pick order ($1,2,\ldots n, n, n-1, \ldots, 1$).
    \item \emph{Random EF1 Mask}: An ablation baseline that selects a feasible recipient uniformly at random from $A_t^{\mathrm{feas}}$ at each step, isolating the welfare contribution of learned priorities from feasibility masking.
    \item \emph{Random RBS}: A baseline enforcing round-level item count balance ($\Delta\text{Count} \le 1$) via random agent permutations per round.
\end{enumerate}
Table~\ref{tab:main_benchmark} summarizes the comparative performance across all scale regimes.

\begin{table}[!h]
\centering
\footnotesize
\renewcommand{\arraystretch}{1.08}
\setlength{\tabcolsep}{6pt}
\begin{tabular}{llcccc}
\toprule
\textbf{\shortstack[l]{Scale\\Regime}} & 
\textbf{\shortstack[l]{Method /\\Baseline}} & 
\textbf{\shortstack{Mean\\NSW}} & 
\textbf{\shortstack{vs. LPT\\Win (\%)}} & 
\textbf{\shortstack{vs. LPT\\Tie (\%)}} & 
\textbf{\shortstack{vs. LPT\\Lose (\%)}} \\
\midrule
\multirow{5}{*}{\shortstack[l]{\textbf{Overall}\\($n \in [2, 20]$)}} 
 & \textbf{PriorityNet (PPO)} & \textbf{0.9911} & \textbf{53.03\%} & 21.03\% & \textbf{25.93\%} \\
 & Offline LPT (Greedy)       & 0.9911          & ---              & ---     & ---              \\
 & Reverse-RR                 & 0.9819          & 15.53\%          & 6.67\%  & 77.80\%          \\
 & Random EF1 Mask            & 0.9774          & 4.77\%           & 1.00\%  & 94.23\%          \\
 & Random RBS                 & 0.9784          & 6.13\%           & 1.23\%  & 92.63\%          \\
\midrule
\multirow{5}{*}{\shortstack[l]{\textbf{Small}\\($n \in [2, 5]$)}} 
 & \textbf{PriorityNet (PPO)} & \textbf{0.9788} & \textbf{24.00\%} & 61.00\% & \textbf{15.00\%} \\
 & Offline LPT (Greedy)       & 0.9787          & ---              & ---     & ---              \\
 & Reverse-RR                 & 0.9683          & 11.80\%          & 19.30\% & 68.90\%          \\
 & Random EF1 Mask            & 0.9619          & 6.70\%           & 3.00\%  & 90.30\%          \\
 & Random RBS                 & 0.9620          & 9.40\%           & 3.70\%  & 86.90\%          \\
\midrule
\multirow{5}{*}{\shortstack[l]{\textbf{Medium}\\($n \in [6, 10]$)}} 
 & \textbf{PriorityNet (PPO)} & \textbf{0.9963} & \textbf{68.50\%} & 1.90\%  & \textbf{29.60\%} \\
 & Offline LPT (Greedy)       & 0.9961          & ---              & ---     & ---              \\
 & Reverse-RR                 & 0.9869          & 17.00\%          & 0.60\%  & 82.40\%          \\
 & Random EF1 Mask            & 0.9829          & 4.50\%           & 0.00\%  & 95.50\%          \\
 & Random RBS                 & 0.9845          & 5.00\%           & 0.00\%  & 95.00\%          \\
\midrule
\multirow{5}{*}{\shortstack[l]{\textbf{Large}\\($n \in [11, 20]$)}} 
 & \textbf{PriorityNet (PPO)} & \textbf{0.9983} & \textbf{66.60\%} & 0.20\%  & \textbf{33.20\%} \\
 & Offline LPT (Greedy)       & 0.9984          & ---              & ---     & ---              \\
 & Reverse-RR                 & 0.9904          & 17.80\%          & 0.10\%  & 82.10\%          \\
 & Random EF1 Mask            & 0.9874          & 3.10\%           & 0.00\%  & 96.90\%          \\
 & Random RBS                 & 0.9887          & 4.00\%           & 0.00\%  & 96.00\%          \\
\bottomrule
\end{tabular}
\vspace{4pt}
\caption{Offline Benchmark Results: PriorityNet (PPO) vs. Baselines across 3,000 independent test instances ($n \in [2, 20], m \in [5, 100]$).}
\label{tab:main_benchmark}
\end{table}

\vspace{-12pt}
\paragraph{Performance Analysis and Ablation Insights.}
As shown in Table~2, PriorityNet performs competitively with Offline LPT and attains higher mean normalized $\operatorname{NSW}$ than the other considered baselines. Across all 3,000 test instances, PriorityNet and Offline LPT both achieve an overall mean normalized $\operatorname{NSW}$ of $0.9911$. Nevertheless, PriorityNet attains higher $\operatorname{NSW}$ than LPT on $53.03\%$ of the instances, ties on $21.03\%$, and attains lower $\operatorname{NSW}$ on $25.93\%$, yielding a win-minus-loss rate of $+27.10\%$ and a non-loss rate of $74.06\%$. Across the Small, Medium, and Large regimes, the respective mean normalized $\operatorname{NSW}$ values of PriorityNet and LPT are $(0.9788,0.9787)$, $(0.9963,0.9961)$, and $(0.9983,0.9984)$. Thus, PriorityNet has slightly higher mean welfare in the Small and Medium regimes and slightly lower mean welfare in the Large regime, although it wins more often than it loses in all three regimes. In particular, its win rates increase to $68.50\%$ and $66.60\%$ in the Medium and Large regimes, respectively. The Random EF1 Mask ablation attains an overall mean normalized $\operatorname{NSW}$ of $0.9774$ and, relative to LPT, records a $4.77\%$ win rate and a $94.23\%$ loss rate. Its lower performance is consistent with the importance of learned action selection beyond feasibility masking alone, although a direct paired comparison between PriorityNet and Random EF1 Mask would be needed to quantify this contribution more precisely. The round-robin baselines also perform less favorably than LPT: Reverse-RR and Random RBS attain overall mean normalized $\operatorname{NSW}$ values of $0.9819$ and $0.9784$, respectively, with corresponding loss rates of $77.80\%$ and $92.63\%$ against LPT. Overall, these results show that PriorityNet is competitive with LPT in aggregate welfare and obtains higher welfare on a larger fraction of the tested instances, rather than demonstrating a uniformly large improvement in mean welfare.

\paragraph{Multi-Scale Scaling Dynamics and Theoretical Alignment.}
The multi-scale breakdown reveals consistent theoretical alignment with Theorem~\ref{thm:improved_nsw_small_item} (Small-Item Condition). In Small instances where item counts are small, Offline LPT is frequently near-exact; PriorityNet achieves a 61.00\% Tie Rate and an 85.00\% Non-Loss Rate (24.00\% Win vs. 15.00\% Lose). As problem scale expands to Medium ($m \in [20, 50]$) and Large ($m \in [50, 100]$), individual item values become finer fractions of average share $\mu$, allowing PriorityNet to exploit global combinatorial cross-attention to surpass myopic greedy choices in over two-thirds of all instances. Concurrently, Mean Normalized NSW approaches $0.9983\, \text{MaxNSW}$, directly verifying that empirical welfare scales toward the theoretical continuous optimum as $\epsilon \to 0$.

Fig.~\ref{fig:offline_benchmark}(a) illustrates this normalized NSW convergence toward the continuous optimal ceiling $\text{MaxNSW} = 1.00$ across agent counts $n \in [2, 20]$, corroborating the asymptotic scaling behavior predicted by Theorem~\ref{thm:improved_nsw_small_item}. Fig.~\ref{fig:offline_benchmark}(b) decomposes the detailed head-to-head matchup dynamics against Offline LPT, directly demonstrating how PriorityNet's Net Win Rate ($\text{Win} - \text{Loss} \%$, shaded green) remains strictly positive across $n \ge 4$ and scales up to $+20\%\text{--}+30\%$ in larger problem dimensions.

\begin{figure}[!htbp]
\centering
\includegraphics[width=0.98\textwidth]{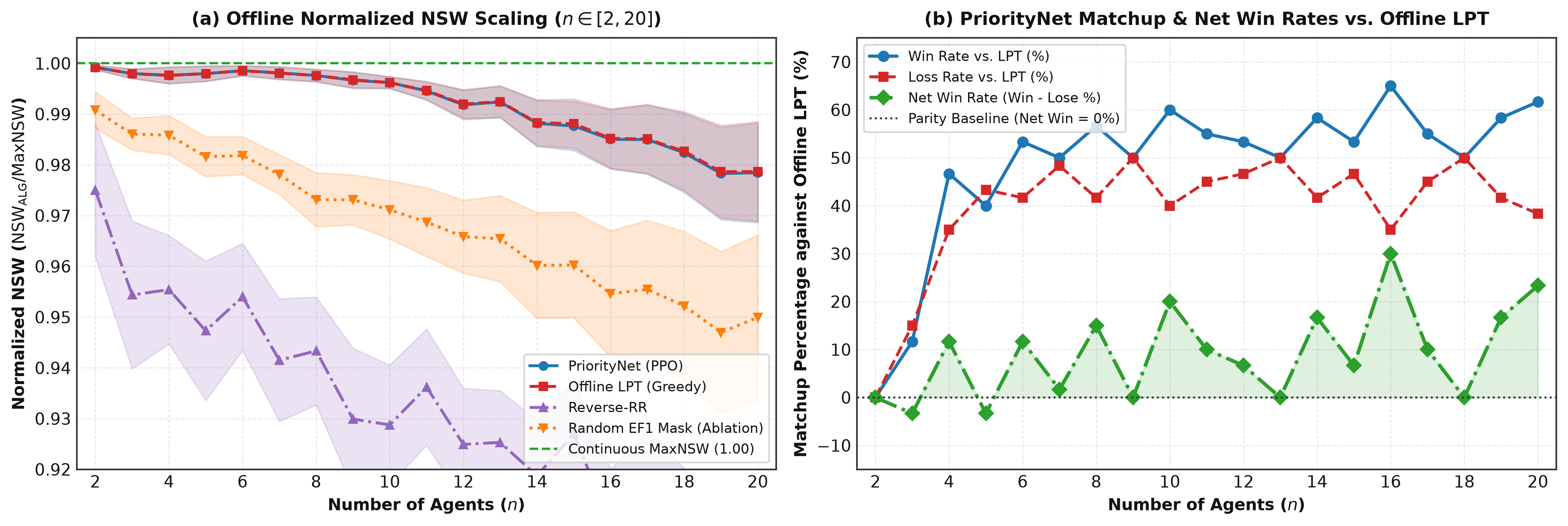}
\caption{Offline full-information scaling dynamics and head-to-head matchup evaluation across agent count $n \in [2, 20]$ (60 random seeds per cell). Panel (a) shows Normalized NSW ($\text{NSW}_{\text{ALG}} / \text{MaxNSW}$) with 95\% Confidence Interval error bands ($1.96 \times \text{SEM}$) converging toward the continuous upper bound $\text{MaxNSW} = 1.00$. Panel (b) illustrates PriorityNet's Win Rate (\%), Loss Rate (\%), and Net Win Rate ($\text{Win} - \text{Loss} \%$, shaded green) against Offline LPT across scaling dimensions.}
\label{fig:offline_benchmark}
\end{figure}

\FloatBarrier

\vspace{-18pt}
\subsubsection{Online Dynamic Streaming Benchmark.}
We evaluate policy performance in the dynamic online streaming setting under the standard \emph{Random-Order Arrival Model} (window $W = 1$), where indivisible goods arrive sequentially in a uniformly random arrival permutation with zero future visibility. To maintain direct symmetry with the offline benchmark, we evaluate the converged online policy $\text{PriorityNet}_{8\text{k}}$ across 3,000 independent test instances over the identical multi-scale regimes: Small ($n \in [2, 5], m \in [5, 20]$), Medium ($n \in [6, 10], m \in [20, 50]$), and Large ($n \in [11, 20], m \in [50, 100]$) drawn across the full spectrum of five distinct valuation distribution profiles.

We benchmark PriorityNet against three online baseline mechanisms:
\begin{enumerate}[leftmargin=*,label=(\alph*)]
    \item \emph{Online LPT (Greedy)}: The online greedy counterpart that immediately assigns each sequentially arriving item to the agent with the lowest accumulated utility ($\argmin_{i} u_i$) without pre-sorting, which strictly preserves prefix-wise EF1 at every arrival step under identical valuations (Theorem~\ref{thm:least-valued-bundle-ef1}).
    \item \emph{Online Reverse-RR}: The non-adaptive fair division protocol assigning arriving items in alternating snake order ($1,2, \ldots n, n, n-1, \ldots, 1$).    Online Reverse-RR is included as an unconstrained non-adaptive reference and does not generally guarantee prefix-wise EF1
    
    \item \emph{Random EF1 Mask}: The neural ablation baseline selecting actions uniformly at random from $A_t^{\mathrm{feas}}$ at each dynamic arrival step.
\end{enumerate}
Table~\ref{tab:online_scaling} summarizes the comparative performance across all scale regimes.

\begin{table}[!htbp]
\centering
\footnotesize
\renewcommand{\arraystretch}{1.05}
\setlength{\tabcolsep}{6pt}
\begin{tabular}{llcccc}
\toprule
\textbf{\shortstack[l]{Scale\\Regime}} & 
\textbf{\shortstack[l]{Method /\\Baseline}} & 
\textbf{\shortstack{Mean\\NSW}} & 
\textbf{\shortstack{vs. LPT\\Win (\%)}} & 
\textbf{\shortstack{vs. LPT\\Tie (\%)}} & 
\textbf{\shortstack{vs. LPT\\Lose (\%)}} \\
\midrule
\multirow{4}{*}{\shortstack[l]{\textbf{Overall}\\($n \in [2, 20]$)}} 
 & \textbf{PriorityNet (PPO)} & \textbf{0.9701} & \textbf{54.27\%} & 9.33\%  & \textbf{36.40\%} \\
 & Online LPT (Greedy)        & 0.9694          & ---              & ---     & ---              \\
 & Online Reverse-RR          & 0.9005          & 8.43\%           & 1.97\%  & 89.60\%          \\
 & Random EF1 Mask            & 0.9608          & 27.13\%          & 5.60\%  & 67.27\%          \\
\midrule
\multirow{4}{*}{\shortstack[l]{\textbf{Small}\\($n \in [2, 5]$)}} 
 & \textbf{PriorityNet (PPO)} & \textbf{0.9610} & \textbf{37.00\%} & 28.00\% & \textbf{35.00\%} \\
 & Online LPT (Greedy)        & 0.9617          & ---              & ---     & ---              \\
 & Online Reverse-RR          & 0.9165          & 18.30\%          & 5.90\%  & 75.80\%          \\
 & Random EF1 Mask            & 0.9531          & 26.10\%          & 16.80\% & 57.10\%          \\
\midrule
\multirow{4}{*}{\shortstack[l]{\textbf{Medium}\\($n \in [6, 10]$)}} 
 & \textbf{PriorityNet (PPO)} & \textbf{0.9765} & \textbf{58.10\%} & 0.00\%  & \textbf{41.90\%} \\
 & Online LPT (Greedy)        & 0.9759          & ---              & ---     & ---              \\
 & Online Reverse-RR          & 0.9039          & 5.10\%           & 0.00\%  & 94.90\%          \\
 & Random EF1 Mask            & 0.9676          & 30.50\%          & 0.00\%  & 69.50\%          \\
\midrule
\multirow{4}{*}{\shortstack[l]{\textbf{Large}\\($n \in [11, 20]$)}} 
 & \textbf{PriorityNet (PPO)} & \textbf{0.9726} & \textbf{67.70\%} & 0.00\%  & \textbf{32.30\%} \\
 & Online LPT (Greedy)        & 0.9705          & ---              & ---     & ---              \\
 & Online Reverse-RR          & 0.8811          & 1.90\%           & 0.00\%  & 98.10\%          \\
 & Random EF1 Mask            & 0.9616          & 24.80\%          & 0.00\%  & 75.20\%          \\
\bottomrule
\end{tabular}
\vspace{4pt}
\caption{Online Dynamic Streaming Benchmark under the Random-Order Arrival Model: PriorityNet ($\text{PriorityNet}_{8\text{k}}$) vs. Online Baselines across 3,000 independent test instances ($n \in [2, 20], m \in [5, 100]$).}
\label{tab:online_scaling}
\end{table}

\vspace{-18pt}
\paragraph{Performance Analysis across Online Regimes.}
As shown in Table~\ref{tab:online_scaling}, PriorityNet performs competitively with the online LPT baseline under random-order arrivals. Across all 3,000 test instances, PriorityNet attains a mean normalized $\operatorname{NSW}$ of $0.9701$, compared with $0.9694$ for the baseline. PriorityNet obtains higher $\operatorname{NSW}$ on $54.27\%$ of the instances, ties on $9.33\%$, and obtains lower $\operatorname{NSW}$ on $36.40\%$, yielding a win-minus-loss rate of $+17.87\%$ and a non-loss rate of $63.60\%$. Thus, PriorityNet wins on a larger fraction of the tested instances, although the improvement in overall mean normalized welfare is modest.

Across the Small, Medium, and Large regimes, the respective mean normalized $\operatorname{NSW}$ values of PriorityNet and the LPT baseline are~$(0.9610,0.9617)$, $(0.9765,0.9759)$, and $(0.9726,0.9705)$. PriorityNet therefore has a slightly lower mean in the Small regime and slightly higher means in the Medium and Large regimes. Its corresponding win-minus-loss rates are $+2.00\%$, $+16.20\%$, and $+35.40\%$, respectively. These results suggest that PriorityNet's instance-wise advantage becomes more frequent as the tested scale increases, although the differences in mean normalized welfare remain relatively small.

The Random EF1 Mask ablation attains an overall mean normalized $\operatorname{NSW}$ of $0.9608$ and, relative to the LPT baseline, records a $27.13\%$ win rate and a~$67.27\%$ loss rate. Its lower performance is consistent with the benefit of learned action selection beyond feasibility masking alone. However, a direct paired comparison between PriorityNet and Random EF1 Mask would be required to quantify this contribution more precisely. Online Reverse-RR performs less favorably, attaining a mean normalized $\operatorname{NSW}$ of~$0.9005$ and losing to the least-valued-bundle baseline on~$89.60\%$ of the instances.

Fig.~3(a) illustrates the normalized-welfare results under the sampled random-order arrivals, with PriorityNet attaining regime-level means between~$0.9610$ and~$0.9765$. Fig.~3(b) presents the corresponding instance-wise win, loss, and win-minus-loss rates against the LPT baseline as the number of agents varies. Together, these results indicate that PriorityNet maintains high normalized welfare and frequently improves upon the greedy baseline, particularly in the larger tested regimes. 

\begin{figure}[!htbp]
\centering
\includegraphics[width=0.98\textwidth]{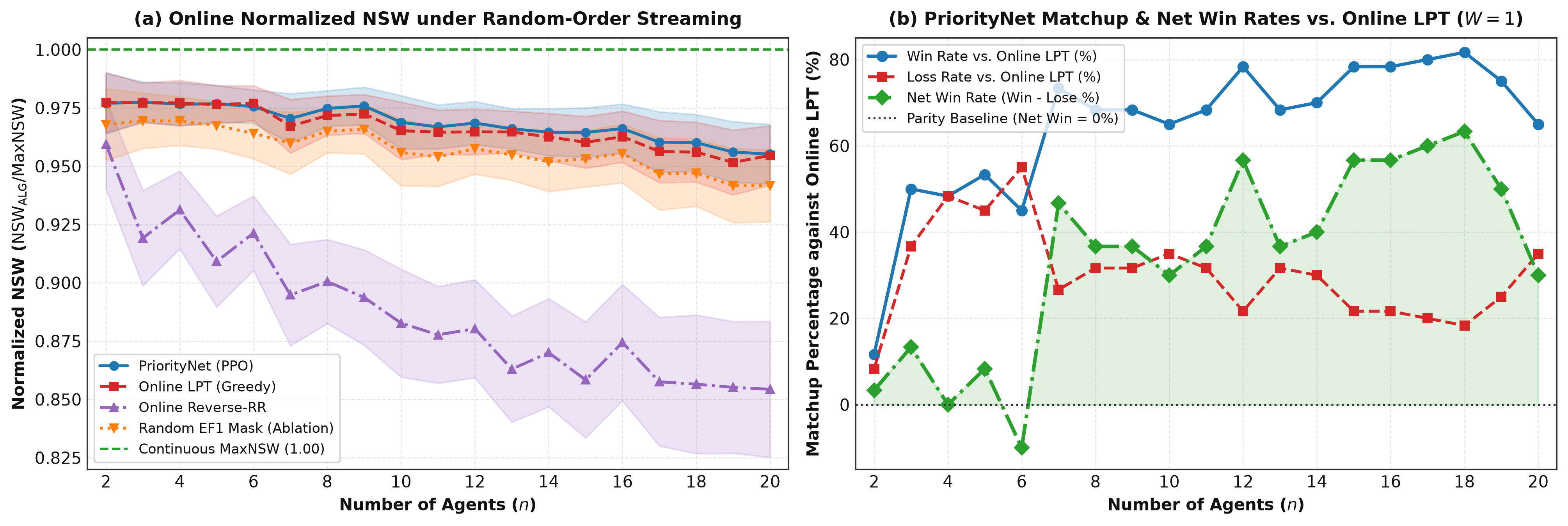}
\caption{Online dynamic streaming benchmark under the Random-Order Arrival Model ($W=1$, 60 random seeds per cell). Panel (a) illustrates Normalized Nash Social Welfare ($\text{NSW}_{\text{ALG}} / \text{MaxNSW}$) with 95\% Confidence Interval error bands under arrival uncertainty. Panel (b) displays PriorityNet's Win Rate (\%), Loss Rate (\%), and Net Win Rate ($\text{Win} - \text{Loss} \%$, shaded green) against Online LPT across agent count $n \in [2, 20]$.}
\label{fig:online_streaming}
\end{figure}

\FloatBarrier

\vspace{-18pt}
\section{Conclusion}
\label{sec:conclusion}

We study $\operatorname{NSW}$ maximization under $\operatorname{EF1}$ for identical additive valuations. Because every maximum-$\operatorname{NSW}$ allocation is $\operatorname{EF1}$, imposing $\operatorname{EF1}$ only on the final allocation leaves the optimum unchanged, and the threshold problem inherits the known strong \textsf{NP}-completeness. We therefore focus on welfare guarantees for arbitrary $\operatorname{EF1}$ allocations. Under uniform valuations, $\operatorname{EF1}$ balances bundle cardinalities and hence every $\operatorname{EF1}$ allocation is $\operatorname{NSW}$-optimal. More generally, under the $\varepsilon$-small-item condition, every $\operatorname{EF1}$ allocation achieves a $\max\{e^{-1/e},\rho_n(\varepsilon)\}$-approximation, where $\rho_n(\varepsilon)=1-O(\varepsilon^2)$ for fixed~$n$. Thus, the guarantee approaches~$1$ as individual goods become negligible relative to the average share.

Separately, we considered the stronger sequential problem in which $\operatorname{EF1}$ must be maintained after every item assignment. For this setting, we introduced PriorityNet, a reinforcement-learning framework trained using Proximal Policy Optimization and equipped with prospective $\operatorname{EF1}$ action masking. The masking mechanism guarantees prefix-wise $\operatorname{EF1}$ by construction for the processed item order. In the offline experiments, PriorityNet and LPT both attained an overall mean normalized $\operatorname{NSW}$ of $0.9911$, although PriorityNet won on more individual instances, yielding a reported win-minus-loss rate of $+27.10\%$. On the large offline subset, PriorityNet attained a mean normalized $\operatorname{NSW}$ of $0.9983$, compared with $0.9984$ for LPT. In the online streaming experiments, PriorityNet attained an overall mean normalized $\operatorname{NSW}$ of $0.9701$, compared with $0.9694$ for the online least-valued-bundle baseline, together with a win-minus-loss rate of $+17.87\%$. These results indicate that PriorityNet is competitive with the considered heuristics and wins on a larger fraction of the tested instances, although the differences in aggregate mean welfare are modest.

Promising avenues for future research include extending the prospective priority field architecture to asymmetric agent valuations and general submodular valuation profiles, investigating dynamic streaming with finite lookahead buffer windows ($W > 1$), and establishing theoretical sample-complexity and generalization bounds for reinforcement learning policies in constrained fair division.

\begin{credits}
\subsubsection{\ackname} This work is funded and supported by the National Science and Technology Council, Taiwan, under grant nos. NSTC 115-2221-E-019-045-.
\end{credits}

\bibliographystyle{splncs04}
\bibliography{max_nsw_ef1}

\FloatBarrier
\newpage
\section*{Appendix}                                                   
\appendix
\renewcommand*{\theHsection}{appendix.\Alph{section}}                 
\renewcommand*{\theHsubsection}
{\theHsection.\arabic{subsection}}                                                                
 \renewcommand*{\theHsubsubsection}{\theHsubsection.\arabic{subsubsectio 
 n}}
\section{Masked Categorical Policy and PPO Likelihood Computation}
\label{app:masked_policy}

This appendix provides the complete formulation of the prospective EF1
action mask, the resulting masked Categorical policy, and the corresponding
likelihood-ratio computation used during PPO optimization.

\subsection{Prospective EF1 Action Mask}
\label{app:prospective_ef1_mask}

At allocation step $t\in\{1,2,\ldots,m\}$, the Actor produces a vector of raw
priority logits
\begin{equation}
\boldsymbol{\ell}_t
=
(\ell_{t,1},\ell_{t,2},\ldots,\ell_{t,n})
\in \mathbb{R}^n.
\end{equation}

For the current good $g_t$, the prospectively EF1-feasible recipient set is
defined as
\begin{equation}
A_t^{\mathrm{feas}}
=
\left\{
i\in\mathcal{A}
\;\middle|\;
A^{t-1}\cup\{(i,g_t)\}
\text{ is } \operatorname{EF1}
\right\}.
\label{eq:appendix_feasible_set}
\end{equation}

The corresponding binary action mask
$\mathbf{M}_t\in\{0,1\}^n$ is defined componentwise by
\begin{equation}
M_{t,i}
=
\begin{cases}
1, & i\in A_t^{\mathrm{feas}},\\
0, & i\notin A_t^{\mathrm{feas}}.
\end{cases}
\label{eq:appendix_binary_mask}
\end{equation}

The mask is applied directly to the Actor logits:
\begin{equation}
\widetilde{\ell}_{t,i}
=
\begin{cases}
\ell_{t,i},
& i\in A_t^{\mathrm{feas}},\\
-\infty,
& i\notin A_t^{\mathrm{feas}}.
\end{cases}
\label{eq:masked_logits}
\end{equation}

Thus, infeasible actions are removed from the support of the policy before
probability normalization.

Under identical nonnegative additive valuations,
Theorem~\ref{thm:least-valued-bundle-ef1} guarantees that assigning the
arriving good to an agent whose current bundle has minimum value preserves
EF1. Therefore, whenever the current prefix allocation $A^{t-1}$ is EF1,
\begin{equation}
\argmin_{i\in N} v(A_i^{t-1})
\subseteq
A_t^{\mathrm{feas}}.
\label{eq:least_bundle_feasible_subset}
\end{equation}

Since a minimum-valued bundle always exists,
\begin{equation}
A_t^{\mathrm{feas}}
\neq
\emptyset
\qquad
\text{for every } t\in\{1,2,\ldots,m\}.
\label{eq:nonempty_feasible_actions}
\end{equation}

The empty initial allocation is trivially EF1, and every subsequent action is
restricted to $A_t^{\mathrm{feas}}$. Hence, EF1 is preserved throughout the
entire allocation sequence. Moreover,
Eq.~\eqref{eq:nonempty_feasible_actions} guarantees that at least one
admissible recipient is available at every step, so no fallback allocation
rule is required.

\subsection{Masked Categorical Policy}
\label{app:masked_categorical_policy}

Applying softmax to the masked logits in
Eq.~\eqref{eq:masked_logits} yields the constrained policy
\begin{equation}
\pi_{\theta}^{\mathbf{M}_t}(a_t=i\mid s_t)
=
\frac{
    \exp\!\left(\widetilde{\ell}_{t,i}\right)
}{
    \displaystyle
    \sum_{j=1}^{n}
    \exp\!\left(\widetilde{\ell}_{t,j}\right)
}.
\label{eq:masked_policy_softmax}
\end{equation}

Equivalently, the masked Categorical distribution can be written explicitly
as
\begin{equation}
\pi_{\theta}^{\mathbf{M}_t}(a_t=i\mid s_t)
=
\begin{cases}
\displaystyle
\frac{
    \exp(\ell_{t,i})
}{
    \sum_{j\in A_t^{\mathrm{feas}}}
    \exp(\ell_{t,j})
},
&
i\in A_t^{\mathrm{feas}},
\\[10pt]
0,
&
i\notin A_t^{\mathrm{feas}}.
\end{cases}
\label{eq:masked_categorical_policy}
\end{equation}

Thus, the Actor produces scores for all agents, while probability
normalization is performed only over the state-dependent feasible action set
$A_t^{\mathrm{feas}}$. In particular, every EF1-infeasible action receives
zero probability.

During training rollouts, the recipient agent is sampled from the masked
Categorical policy:
\begin{equation}
a_t
\sim
\pi_{\theta}^{\mathbf{M}_t}(\cdot\mid s_t).
\label{eq:masked_training_sampling}
\end{equation}

This provides stochastic exploration among alternative allocations while
restricting exploration entirely to EF1-feasible recipients.

During deterministic evaluation and inference, sampling is replaced by
selection of the highest-priority feasible agent:
\begin{equation}
a_t
=
\argmax_{i\in A_t^{\mathrm{feas}}}
\ell_{t,i}.
\label{eq:masked_deterministic_action}
\end{equation}

Because every infeasible logit is replaced by $-\infty$ before action
selection, Eq.~\eqref{eq:masked_deterministic_action} is equivalently obtained
by taking the $\argmax$ over the masked logits
$\widetilde{\boldsymbol{\ell}}_t$.

\subsection{Masked PPO Likelihood Ratio}
\label{app:masked_ppo_ratio}

The feasible action set depends on the current allocation state. Consequently,
the same state-dependent mask used to generate a rollout action must also be
used when evaluating that action during PPO optimization.

For each rollout transition, the mask $\mathbf{M}_t$ associated with state
$s_t$ is retained together with the rollout information and reapplied when
the sampled action is evaluated under the updated policy. The behavior-policy
log-probability is therefore
\begin{equation}
\log
\pi_{\theta_{\mathrm{old}}}^{\mathbf{M}_t}
(a_t\mid s_t),
\label{eq:old_masked_logprob}
\end{equation}
whereas the corresponding log-probability under the updated policy is
\begin{equation}
\log
\pi_{\theta}^{\mathbf{M}_t}
(a_t\mid s_t).
\label{eq:new_masked_logprob}
\end{equation}

The PPO likelihood ratio is consequently defined as
\begin{equation}
r_t(\theta)
=
\frac{
    \pi_{\theta}^{\mathbf{M}_t}(a_t\mid s_t)
}{
    \pi_{\theta_{\mathrm{old}}}^{\mathbf{M}_t}(a_t\mid s_t)
}.
\label{eq:masked_ppo_ratio}
\end{equation}

Equivalently, using stored log-probabilities,
\begin{equation}
r_t(\theta)
=
\exp\!\left(
    \log
    \pi_{\theta}^{\mathbf{M}_t}(a_t\mid s_t)
    -
    \log
    \pi_{\theta_{\mathrm{old}}}^{\mathbf{M}_t}(a_t\mid s_t)
\right).
\label{eq:masked_ppo_ratio_log}
\end{equation}

The use of the mask in both terms is important because the rollout action was
sampled from the constrained policy rather than from the unconstrained Actor
distribution. To make this distinction explicit, let
$\pi_{\theta}^{\mathrm{raw}}$ denote the categorical distribution obtained by
applying softmax directly to the unmasked logits. For a feasible action
$i\in A_t^{\mathrm{feas}}$, the masked probability can be expressed as
\begin{equation}
\pi_{\theta}^{\mathbf{M}_t}(a_t=i\mid s_t)
=
\frac{
    \pi_{\theta}^{\mathrm{raw}}(a_t=i\mid s_t)
}{
    Z_{\theta}(s_t,\mathbf{M}_t)
},
\label{eq:masked_raw_relation}
\end{equation}
where
\begin{equation}
Z_{\theta}(s_t,\mathbf{M}_t)
=
\sum_{j\in A_t^{\mathrm{feas}}}
\pi_{\theta}^{\mathrm{raw}}(a_t=j\mid s_t)
\label{eq:feasible_probability_mass}
\end{equation}
is the total probability mass that the unconstrained policy assigns to
EF1-feasible actions.

Substituting Eq.~\eqref{eq:masked_raw_relation} into
Eq.~\eqref{eq:masked_ppo_ratio} gives
\begin{equation}
r_t(\theta)
=
\frac{
    \pi_{\theta}^{\mathrm{raw}}(a_t\mid s_t)
}{
    \pi_{\theta_{\mathrm{old}}}^{\mathrm{raw}}(a_t\mid s_t)
}
\cdot
\frac{
    Z_{\theta_{\mathrm{old}}}(s_t,\mathbf{M}_t)
}{
    Z_{\theta}(s_t,\mathbf{M}_t)
}.
\label{eq:masked_vs_unmasked_ratio}
\end{equation}

Equation~\eqref{eq:masked_vs_unmasked_ratio} shows that the likelihood ratio
of the masked policy is generally not equal to the likelihood ratio computed
from the unmasked Actor distributions. In particular, the updated and
behavior policies may assign different total probability mass to infeasible
actions, so that
\begin{equation}
Z_{\theta}(s_t,\mathbf{M}_t)
\neq
Z_{\theta_{\mathrm{old}}}(s_t,\mathbf{M}_t)
\end{equation}
in general.

Therefore, computing the PPO ratio from unmasked probabilities would evaluate
a policy different from the constrained policy that generated the rollout.
Reapplying the same mask $\mathbf{M}_t$ ensures that both the numerator and
denominator of the PPO likelihood ratio are normalized over the identical
EF1-feasible action set. PPO updates consequently modify the relative
preferences among feasible recipients while EF1-infeasible actions remain
outside the support of the policy.

\section{PPO Objective and Advantage Estimation}
\label{app:ppo_objective}

\subsection{Generalized Advantage Estimation}

Temporal advantages are estimated using Generalized Advantage Estimation
(GAE-$\lambda$)~\cite{schulman2015high}.
Let $t \in \{0,\ldots,T-1\}$ denote the current decision step in a trajectory
of length $T$. At step $t$, $s_t$ denotes the environment state and $r_t$
denotes the scalar reward. The temporal-difference (TD) residual is defined as
\begin{equation}
\delta_t
=
r_t + \gamma V(s_{t+1}) - V(s_t),
\label{eq:appendix_td_residual}
\end{equation}
where $\delta_t$ is the one-step TD residual, $V(s_t)$ is the Critic estimate
of the expected discounted return from state $s_t$, and
$\gamma \in [0,1]$ is the discount factor. For a terminal state, the bootstrap
value is set to zero, i.e., $V(s_T)=0$.

The corresponding GAE estimate is
\begin{equation}
A_t
=
\sum_{l=0}^{T-t-1}
(\gamma \lambda)^l \delta_{t+l},
\label{eq:appendix_gae}
\end{equation}
where $A_t$ denotes the estimated advantage at step $t$, $l$ is the temporal
offset from the current step, and $\lambda \in [0,1]$ is the GAE parameter
controlling the bias--variance trade-off. Smaller values of $\lambda$ place
greater weight on short-horizon bootstrapped estimates, whereas larger values
incorporate TD residuals over longer horizons.

The value-function regression target is constructed from the unnormalized
advantage estimate as
\begin{equation}
\hat{R}_t
=
A_t + V(s_t),
\label{eq:appendix_return_target}
\end{equation}
where $\hat{R}_t$ denotes the target discounted return used to train the
Critic.

The Critic is optimized using the squared-error objective
\begin{equation}
\mathcal{L}_{\mathrm{value}}(\theta)
=
\frac{1}{2}
\hat{\mathbb{E}}_t
\left[
    \left(
        V_\theta(s_t) - \hat{R}_t
    \right)^2
\right],
\label{eq:appendix_value_loss}
\end{equation}
where $\mathcal{L}_{\mathrm{value}}(\theta)$ denotes the Critic loss,
$\theta$ denotes the trainable Actor--Critic network parameters,
$V_\theta(s_t)$ denotes the Critic value prediction under parameters
$\theta$, and $\hat{\mathbb{E}}_t[\cdot]$ denotes the empirical average over
valid sampled time steps in the current mini-batch. The regression target
$\hat{R}_t$ is treated as fixed during the corresponding Critic optimization
step.

\subsection{Advantage Normalization}

For policy optimization, the advantages are standardized within each
mini-batch:
\begin{equation}
\bar{A}_t
=
\frac{A_t - \mu_A}
     {\sqrt{\sigma_A^2 + 10^{-8}}},
\label{eq:appendix_advantage_normalization}
\end{equation}
where $\bar{A}_t$ denotes the normalized advantage, $\mu_A$ denotes the
empirical mean of the valid advantage estimates in the current mini-batch,
and $\sigma_A^2$ denotes their empirical variance. The constant $10^{-8}$
is included for numerical stability.

Advantage normalization is applied only to the policy surrogate objective.
The original, unnormalized advantages $A_t$ are retained when constructing
the Critic targets in Eq.~\eqref{eq:appendix_return_target}.

\subsection{Clipped PPO Surrogate Objective}

Because actions are sampled from the EF1-constrained policy, the PPO
likelihood ratio is evaluated using the same state-dependent mask
$\mathbf{M}_t$:
\begin{equation}
r_t(\theta)
=
\frac{
\pi_{\theta}^{\mathbf{M}_t}(a_t\mid s_t)
}{
\pi_{\theta_{\mathrm{old}}}^{\mathbf{M}_t}(a_t\mid s_t)
}.
\label{eq:appendix_ppo_ratio}
\end{equation}

The clipped surrogate loss is then
\begin{equation}
\mathcal{L}_{\mathrm{clip}}(\theta)
=
-
\hat{\mathbb{E}}_t
\left[
\min
\left(
r_t(\theta)\bar{A}_t,
\operatorname{clip}
\left(
r_t(\theta),
1-\epsilon,
1+\epsilon
\right)
\bar{A}_t
\right)
\right],
\label{eq:appendix_ppo_clip}
\end{equation}
where $\epsilon$ is the PPO clipping coefficient.

The complete Actor--Critic objective is
\begin{equation}
\mathcal{L}_{\mathrm{total}}(\theta)
=
\mathcal{L}_{\mathrm{clip}}(\theta)
+
c_{vf}\mathcal{L}_{\mathrm{value}}(\theta)
-
c_{\mathrm{ent}}
\mathcal{H}
\left(
\pi_{\theta}^{\mathbf{M}_t}
\right).
\label{eq:appendix_total_ppo_loss}
\end{equation}

Here, $c_{vf}$ controls the contribution of the Critic regression loss, while
$c_{\mathrm{ent}}$ controls entropy regularization and encourages sufficient
exploration among EF1-feasible actions.

\subsection{Optimization Details.}
The Actor--Critic network is jointly optimized using the Adam optimizer~\cite{kingma2014adam}. For each PPO update, the collected rollout samples are randomly shuffled at the beginning of every optimization epoch and partitioned into mini-batches. Each mini-batch is then used to compute the clipped policy loss, value-function loss, and entropy regularization term, followed by an Adam parameter update. The same rollout data are reused for multiple PPO optimization epochs, with a new random permutation generated at each epoch. The learning rate, PPO clipping coefficient, GAE parameters, value-loss coefficient, entropy coefficient, mini-batch size, and number of PPO optimization epochs are reported in Table~\ref{tab:ppo_hyperparameters}. 

\section{Lemmas and Their Proofs}

\begin{lemma}[EF1 implies bounded utility spread]
\label{lem:small-item-spread}
Let $A$ be an EF1 allocation under identical additive valuations.  Then for all agents $i,j$,
\[
    v(A_j)-v(A_i)\leq v^*.
\]
Equivalently, the largest and smallest bundle values differ by at most~$v^*$.
\end{lemma}

\begin{proof}
Fix $i,j\in \{1,2,\ldots,n\}$.  If $A_j=\emptyset$, the inequality is immediate.  Otherwise, since $A$ is EF1, there exists $g\in A_j$ such that
$v(A_i)\geq v(A_j\setminus\{g\})=v(A_j)-v(g)$. Because $v(g)\leq v^*$, we have $v(A_i)\geq v(A_j)-v^*$, which is equivalent to $v(A_j)-v(A_i)\leq v^*$.
\qed
\end{proof}

The next lemma is the product estimate used for the small-item guarantee.  The minimum over $k$ is important in the sense that the extremal product needs not occur when exactly one bundle is poor and all other bundles are rich.

\begin{lemma}[Product lower bound under bounded spread]
\label{lem:bounded_spread}
Let $n\geq 2$ and let $x_1,x_2,\ldots,x_n>0$ satisfy
\[
    \sum_{i=1}^n x_i=n\mu
    \quad\text{and}\quad
    \max_{i\in\{1,2,\ldots,n\}} x_i-\min_{j\in\{1,2,\ldots,n\}} x_j\leq D,
\]
where $0\leq D\leq \mu$.  Then
\[
    \prod_{i=1}^n x_i
    \geq
    \mu^n
    \min_{1\le k\le n-1}
    \left(1-\frac{kD}{n\mu}\right)^{n-k}
    \left(1+\frac{(n-k)D}{n\mu}\right)^k .
\]
\end{lemma}

\begin{proof}
Let
\[
    d := \max_{i\in\{1,2,\ldots,n\}} x_i-\min_{j\in\{1,2,\ldots,n\}} x_j.
\]
We first prove the structure of a product-minimizing vector with fixed
sum and spread at most \(D\).
Consider a feasible vector minimizing \(\prod_{i=1}^n x_i\). Such a
minimizer exists because the constraints
\[
    \sum_{i=1}^n x_i=n\mu,
    \qquad
    x_i-x_j\le D \quad \text{for all } i,j\in\{1,2,\ldots,n\}
\]
define a closed and bounded feasible region. Moreover, since \(D\leq \mu\),
all feasible coordinates are nonnegative, and in fact positive unless all
coordinates are equal.

Let \(M=\max_{i\in\{1,2,\ldots,n\}} x_i\) and \(m=\min_{j\in\{1,2,\ldots,n\}} x_i\), so that \(M-m=d\leq D\).
Suppose two coordinates \(x_p\le x_q\) both lie strictly between \(m\) and
\(M\). For sufficiently small \(t>0\), replacing them by $x_p-t$ and $x_q+t$, respectively, 
preserves the total sum and keeps all coordinates in the interval
\([m,M]\), and hence preserves feasibility. However, the product of these
two coordinates changes from \(x_px_q\) to
\[
    (x_p-t)(x_q+t)
    =
    x_px_q+t(x_p-x_q)-t^2
    <
    x_px_q.
\]
All other coordinates are unchanged, so the total product strictly
decreases. This contradicts the minimality of the vector. Therefore, at a
product minimizer, at most one coordinate can lie strictly between the
minimum value \(m\) and the maximum value~\(M\).

We next show that the possible interior coordinate can also be eliminated.
Fix the spread \(d=M-m\) and fix the number \(k\) of coordinates equal to
the upper endpoint. Suppose there is one interior coordinate. Write the
lower endpoint as \(a\), the upper endpoint as \(a+d\), and the interior
coordinate as \(a+t\), where \(0<t<d\). Then the sum constraint gives
$(n-k-1)a+(a+t)+k(a+d)=n\mu$, so
\[
    a(t)=\mu-\frac{kd+t}{n}.
\]
For this fixed \(k\) and \(d\), the logarithm of the product is
\[
    \phi(t)
    =
    (n-k-1)\log\!\left(\mu-\frac{kd+t}{n}\right) + \log\!\left(\mu-\frac{kd+t}{n}+t\right) + k\log\!\left(\mu-\frac{kd+t}{n}+d\right).
\]
Each term is the logarithm of an affine function of \(t\), and hence
\(\phi(t)\) is concave on its feasible interval. A concave function on an
interval attains its minimum at an endpoint. Thus the product cannot be
increased by replacing the interior coordinate by one of the two endpoints.
Consequently, a product-minimizing vector may be chosen with all
coordinates equal either to \(a\) or to \(a+d\).

Hence, for some \(k\in\{1,2,\ldots,n-1\}\), the extremal vector has
\(n-k\) coordinates equal to \(a\) and \(k\) coordinates equal to \(a+d\).
The sum constraint gives
\[
    (n-k)a+k(a+d)=n\mu,
\]
and therefore
\[
    a=\mu-\frac{kd}{n},
    \quad
    a+d=\mu+\frac{(n-k)d}{n}.
\]
Thus the product is at least
\[
    \left(\mu-\frac{kd}{n}\right)^{n-k}
    \left(\mu+\frac{(n-k)d}{n}\right)^k
\]
for some \(k\in\{1,2,\ldots,n-1\}\).

It remains to replace \(d\) by the upper bound \(D\). For fixed
\(k\in\{1,2,\ldots,n-1\}\), define
\[
    F_k(d)
    =
    \left(\mu-\frac{kd}{n}\right)^{n-k}
    \left(\mu+\frac{(n-k)d}{n}\right)^k .
\]
A direct calculation gives
\[
    \odv{}{d}\log F_k(d) = -\frac{k(n-k)}{n} \frac{1}{\mu-\frac{kd}{n}}
    + \frac{k(n-k)}{n} \frac{1}{\mu+\frac{(n-k)d}{n}} <\, 0
\]
whenever \(d>0\). Hence \(F_k(d)\) is nonincreasing in \(d\). Since
\(d\le D\), we have
\[
    F_k(d)\ge F_k(D).
\]
Therefore
\[
    \prod_{i=1}^n x_i
    \geq
    \min_{1\le k\le n-1}
    \left(\mu-\frac{kD}{n}\right)^{n-k}
    \left(\mu+\frac{(n-k)D}{n}\right)^k .
\]
Equivalently,
\[
    \prod_{i=1}^n x_i
    \geq
    \mu^n
    \min_{1\le k\le n-1}
    \left(1-\frac{kD}{n\mu}\right)^{n-k}
    \left(1+\frac{(n-k)D}{n\mu}\right)^k .
\]
Thus the claimed bound follows.
\qed
\end{proof}

\section{A Self-Contained Proof of Theorem~\ref{thm:NPC}}
\label{appendix:proof_theorem_1}

\begin{proof}
\emph{\textsf{NP}-membership.} We first show that the problem belongs to \textsf{NP}. 
A certificate is an allocation $A=(A_1,A_2,\ldots,A_n)$,
which can be encoded by specifying, for each good, the agent to whom it is assigned.
The size of this certificate is polynomial in the input size.
Given such an allocation, we can verify in polynomial time that it is a valid allocation,
namely that every good is assigned to exactly one agent. We then compute the bundle
values $x_i := v(A_i)=\sum_{g\in A_i} v(g)$ for all~$i\in N$.
Since the valuation is additive and the item values are part of the input, all these
values can be computed in polynomial time using integer arithmetic.
Next, we verify the EF1 condition. For every ordered pair of agents \(i,j\), if
\(A_j=\emptyset\), then agent~\(i\) does not envy agent~\(j\). Otherwise, under identical
additive valuations, the condition that there exists a good \(g\in A_j\) such that 
$v(A_i)\geq v(A_j\setminus\{g\})$ is equivalent to $x_i \geq x_j-\max_{g\in A_j} v(g)$.
Thus, for each \(j\), we compute \(\max_{g\in A_j}v(g)\), and then check the above
inequality for all pairs \(i,j\). This requires only polynomial time.
Finally, we verify the Nash social welfare threshold. Since $\NSW(A) = \left(\prod_{i=1}^n x_i\right)^{1/n}$,
the condition \(\operatorname{NSW}(A)\geq T\) is equivalent to $\prod_{i=1}^n x_i \geq T^n$, 
hence we can verify the inequality using integer
arithmetic. The product \(\prod_{i=1}^n x_i\) and the value \(T^n\) have polynomially
many bits, because each \(x_i\) is at most the total value of all goods. Therefore the
threshold condition can also be checked in polynomial time.
Consequently, a proposed allocation can be verified in polynomial time, and the
problem belongs to \textsf{NP}.

\emph{\textsf{NP}-hardness reduction.} We consider a reduction from \textsc{3-Partition}.  An instance of \textsc{3-Partition} consists of positive integers $a_1,a_2,\ldots,a_{3q}$ 
and an integer~$B$ such that $\sum_{j=1}^{3q} a_j=qB \text{ and }\frac{B}{4}<a_j<\frac{B}{2} \text{ for every } j\in\{1,2,\ldots,3q\}$. 
The question is whether the $3q$ numbers can be partitioned into~$q$ triples, each summing exactly to~$B$.
We construct an allocation instance with $q$ agents and $3q$ goods $g_1,g_2,\ldots,g_{3q}$.  All agents have the same additive valuation $v$, defined by
$v(g_j)=a_j$ for each $j=1,2,\ldots,3q$. Set the NSW threshold to $T = B$.
This construction is clearly polynomial.

Suppose first that the \textsc{3-Partition} instance is a yes-instance.  Then the goods can be divided into~$q$ bundles $A_1,A_2,\ldots,A_q$ such that
$v(A_i)=B$ for every $i\in\{1,2,\ldots,q\}$.
The resulting allocation is envy-free, because all agents have identical valuations and all bundles have equal value.  Hence it is EF1.  Moreover,
$\NSW(A)=\left(\prod_{i=1}^q B\right)^{1/q}=B=T$.
Thus the constructed \textsc{EF1-Identical-NSW} instance is a yes-instance.

Conversely, suppose that the constructed allocation instance admits an EF1 allocation $A=(A_1,A_2,\ldots,A_q)$ with $\NSW(A)\geq B$.
Since the valuation is identical and additive, the total value over all bundles is fixed. That is, $\sum_{i=1}^q v(A_i)=v(M)=\sum_{j=1}^{3q} a_j=qB$.
By the arithmetic-geometric mean inequality (AM-GM), $\NSW(A) = \left(\prod_{i=1}^q v(A_i)\right)^{1/q} \leq \frac{1}{q}\sum_{i=1}^q v(A_i) = B$.
Since $\NSW(A)\geq B$, equality must hold in AM-GM.  Therefore $v(A_1)=v(A_2)=\cdots=v(A_q)=B$.
Thus the goods are partitioned into~$q$ bundles, each of total value $B$.
Finally, the size bounds $B/4<a_j<B/2$ imply that each such bundle contains exactly three goods since one or two goods have total value strictly less than~$B$, while four or more goods have total value strictly greater than~$B$.  Hence the bundles define a valid \textsc{3-Partition} solution.

\emph{Concluding.} Since the yes-direction constructs an allocation that is actually envy-free, the EF1 constraint is fully respected by the reduction.  Because \textsc{3-Partition} is strongly \textsf{NP}-complete, the decision problem is strongly \textsf{NP}-complete, and the corresponding optimization problem is thus strongly \textsf{NP}-hard.
\qed
\end{proof}

\section{Proof of Theorem~\ref{thm:improved_nsw_small_item}}
\label{appendix:proof_theorem_4}

\begin{proof}
Let $x_i=v(A_i)$.  Since valuations are identical and additive, $\sum_{i=1}^n x_i=v(M)=n\mu$.
By the bounded-spread lemma, every EF1 allocation satisfies $\max_{i\in\{1,2,\ldots,n\}} x_i-\min_{j\in\{1,2,\ldots,n\}} x_i\leq v^*\leq\eps\mu$. Applying Lemma~\ref{lem:bounded_spread} with $D=\eps\mu$ gives
\[
    \prod_{i=1}^n x_i
    \geq
    \mu^n
    \min_{1\leq k\leq n-1}
    \left(1-\frac{k\eps}{n}\right)^{n-k}
    \left(1+\frac{(n-k)\eps}{n}\right)^k .
\]
Taking the $n$-th root yields $\NSW(A)\ge \mu\rho_n(\eps)$.

On the other hand, for any allocation $B$, the AM-GM inequality gives
\[
    \NSW(B)
    =\left(\prod_{i=1}^n v(B_i)\right)^{1/n}
    \leq
    \frac{1}{n}\sum_{i=1}^n v(B_i)
    =\frac{v(M)}{n}=\mu.
\]
In particular, $\NSW(A^*)\le\mu$.  Hence
\[
    \frac{\NSW(A)}{\NSW(A^*)}
    \geq
    \rho_n(\eps).
\]
The maximum with $e^{-1/e}$ follows because both lower bounds hold simultaneously.
\qed
\end{proof}

\section{State Representation \& PriorityNet Architecture, and Hyperparameter Configuration}

\begin{table}[!h]
\centering
\small
\renewcommand{\arraystretch}{1.45}
\setlength{\extrarowheight}{3pt}
\setlength{\tabcolsep}{8pt}
\begin{tabularx}{\linewidth}{@{} c l l X @{}}
\toprule
\textbf{\shortstack{Feature\\Index}} & \textbf{Symbol} & \textbf{\shortstack[l]{Mathematical\\Definition}} & \textbf{\shortstack[l]{Semantic Description \&\\Interpretation}} \\
\midrule
0 & $\bar{u}_i$ & $\displaystyle \frac{u_i}{\max(U_t^{\max}, 1.0)}$ & Normalized cumulative utility relative to current maximum utility $U_t^{\max} = \max_k u_k$ \\[5pt]
1 & $\mathrm{rank}(u_i)$ & $\displaystyle \frac{\mathrm{rank}(u_i)}{n - 1}$ & Relative utility rank among all agents ($0.0 = \text{poorest}, 1.0 = \text{richest}$) \\[5pt]
2 & $\hat{m}_{\mathrm{rem}}$ & $\displaystyle \frac{m - t}{m}$ & Fraction of remaining unallocated goods (allocation progress clock) \\[5pt]
3 & $v_{\max}^{\mathcal{W}_t}$ & $\displaystyle \frac{\max_{g \in \mathcal{W}_t} v(g)}{100.0}$ & Maximum item valuation in visible window $\mathcal{W}_t$ ($\equiv v_t / 100.0$ when $|\mathcal{W}_t| = 1$) \\[5pt]
4 & $\bar{v}^{\mathcal{W}_t}$ & $\displaystyle \frac{\mathrm{mean}_{g \in \mathcal{W}_t} v(g)}{100.0}$ & Mean item valuation in visible window $\mathcal{W}_t$ ($\equiv v_t / 100.0$ when $|\mathcal{W}_t| = 1$) \\[5pt]
5 & $\Delta u_i^{\max}$ & $\displaystyle \frac{U_t^{\max} - u_i}{\max(U_t^{\max}, 1.0)}$ & Normalized maximum envy gap relative to the richest agent \\[5pt]
6 & $\phi_i$ & $\displaystyle \frac{u_i}{\sum_k u_k + 10^{-8}}$ & Relative utility share of total realized social welfare \\[5pt]
7 & $\widehat{\mathrm{NSW}}_t$ & $\displaystyle \frac{\mathrm{NSW}_{\mathrm{cur}}}{\max(\mathrm{MaxNSW}_t, 10^{-8})}$ & Running normalized NSW relative to revealed-history ceiling $\mathrm{MaxNSW}_t = \frac{1}{n}\sum_{g \in \mathcal{H}_t} v(g)$ (offline: $\mathcal{H}_t = M$) \\[5pt]
8 & $r_i^{\mathrm{last}}$ & $\displaystyle \mathbb{I}(i = a_{t-1})$ & One-Hot indicator for agent selected at step $t-1$ ($1.0$ if $i = a_{t-1}$, $0.0$ otherwise) \\[5pt]
9 & $\tau_i^{\mathrm{idle}}$ & $\displaystyle \min\!\left(\frac{t - t_i^{\mathrm{last}}}{2n}, 1.0\right)$ & Normalized starvation indicator (idle steps elapsed since agent $i$'s last pick) \\[5pt]
10 & $\Delta v_i^{\mathrm{top}}$ & $\displaystyle \frac{|v_t - \max(A_i)|}{100.0}$ & Normalized value gap between arriving item $v_t$ and agent $i$'s top existing item \\
\bottomrule
\end{tabularx}
\vspace{4pt}
\caption{Engineered observation features.}
\label{tab:state_features}
\end{table}

\FloatBarrier

\begin{table}[H]
\centering
\small
\begin{tabular}{lll}
\toprule
\textbf{Hyperparameter} & \textbf{Symbol / Key} & \textbf{Value} \\
\midrule
Total Training Episodes & \texttt{episodes} & 8,000 \\
Batch Size (Parallel Envs) & $B$ / \texttt{batch\_size} & 64 \\
Learning Rate & \texttt{learning\_rate} & $3 \times 10^{-4}$ (Cosine Decay to $10\%$) \\
PPO Clip Epsilon & $\epsilon$ / \texttt{ppo\_clip\_epsilon} & 0.1 \\
PPO Epochs & \texttt{ppo\_epochs} & 4 \\
PPO Minibatch Size & \texttt{ppo\_minibatch\_size} & 32 \\
GAE Lambda & $\lambda_{\text{gae}}$ & 0.95 \\
Discount Factor & $\gamma$ & 0.995 \\
Value Loss Coefficient & $c_{vf}$ / \texttt{vf\_coef} & 0.5 \\
Entropy Coefficient & $c_{ent}$ / \texttt{entropy\_coef} & 0.02 \\
Gradient Clipping Norm & \texttt{max\_grad\_norm} & 1.0 \\
\bottomrule
\end{tabular}
\vspace{6pt}
\caption{Hyperparameter Configuration for PriorityNet Training.}
\label{tab:ppo_hyperparameters}
\end{table}

\end{document}